\documentclass{new_tlp}

\usepackage{amssymb,amsmath}

\usepackage[polish,english]{babel}
\usepackage[T1,nomathsymbols]{polski}
\usepackage[utf8]{inputenc}

\usepackage{mathtools}
\usepackage{thm-restate}

\usepackage{cleveref}

\usepackage{tikz}
\usetikzlibrary{arrows,backgrounds,automata,topaths,patterns,decorations.pathreplacing,decorations.pathmorphing}
\usetikzlibrary{arrows, chains, positioning, quotes, shapes.geometric}
\usetikzlibrary{calc}
\usetikzlibrary{fit}
\usetikzlibrary{arrows}
\usetikzlibrary{positioning,automata}

\definecolor{Green}{RGB}{0,144,0}

\usepackage{xspace}

\newcommand{\AC}{\ensuremath{\textup{\textsc{AC}}^0}}

\newcommand{\PS}{\ensuremath{\textup{\textsc{PSpace}}}}
\newcommand{\EXPS}{\ensuremath{\textup{\textsc{ExpSpace}}}}

\newcommand{\MTL}{\ensuremath{\textup{DatalogMTL}}}
\newcommand{\MTLneg}{\ensuremath{\textup{DatalogMTL}^{\mkern-2mu \neg}}}

\newcommand{\MTLfp}{\ensuremath{\textup{DatalogMTL}^\neg_{\mathsf{FP}}}}

\newcommand{\matA}{M}

\newcommand{\xbf}{\mathbf{x}}

\newcommand{\cbf}{\mathbf{c}}
\newcommand{\sbf}{\mathbf{s}}

\DeclareFontFamily{U}{MnSymbolC}{}
\DeclareSymbolFont{MnSyC}{U}{MnSymbolC}{m}{n}
\DeclareMathSymbol{\diamondminus}{\mathbin}{MnSyC}{120}
\DeclareMathSymbol{\diamondplus}{\mathbin}{MnSyC}{"7C}
\DeclareMathSymbol{\meddiamond}{\mathbin}{MnSyC}{110}
\DeclareFontShape{U}{MnSymbolC}{m}{n}{
    <-5>  MnSymbolC4
    <5-6>  MnSymbolC5
   <6-7>  MnSymbolC6
   <7-8>  MnSymbolC7
   <8-9>  MnSymbolC8
   <9-10> MnSymbolC9
  <10-12> MnSymbolC10
  <12->   MnSymbolC12}{}
\newcommand{\So}{\mathcal{S}}
\newcommand{\Uo}{\mathcal{U}}

\newcommand{\M}{\mathcal{M}}
\newcommand{\I}{\mathfrak{I}}
\newcommand{\D}{\mathcal{D}}

\newcommand{\PD}{(\Pi,\D)}

\newcommand{\A}{\mathcal{A}}
\newcommand{\B}{\mathcal{B}}
\newcommand{\C}{\mathcal{C}}
\newcommand{\F}{\mathcal{F}}

\newcommand{\Prog}{\Pi}
\newcommand{\TM}{\mathfrak{M}}

\newcommand{\word}{w}
\newcommand{\wordR}{\ensuremath{\overrightarrow{w}}}
\newcommand{\wordL}{\ensuremath{\overleftarrow{w}}}

\newcommand{\Left}{\mathsf{L}}

\newcommand{\Right}{\mathsf{R}}
\newcommand{\states}{\mathcal{Q}}
\newcommand{\qinit}{q_{\textsl{init}}}

\newcommand{\qhalt}{q_{\textsl{halt}}}

\newcommand{\Q}{\mathbb{Q}}
\newcommand{\Z}{\mathbb{Z}}

\newcommand{\T}{\mathbb{T}}

\newcommand{\mat}{\mathsf{at}}

\newcommand{\window}{\mathcal{W}}

\newcommand{\ground}[2]{\mathsf{ground}(#1,#2)}
\newcommand{\groundp}[1]{\mathsf{ground}(#1)}

\newcommand\rxspace{\!\mathop{}\nolimits}

\newcommand{\nott}{\mathsf{not}\rxspace}

\newcommand{\Hmod}{\mathfrak{H}}
\newcommand{\Tmod}{\mathfrak{T}}

\newcommand{\acc}{F}

\usepackage{mathptmx}
\usepackage{enumitem}
 
\newcommand\bcmdtab{\noindent\bgroup\tabcolsep=0pt%
  \begin{tabular}{@{}p{10pc}@{}p{20pc}@{}}}
\newcommand\ecmdtab{\end{tabular}\egroup}

  \title[DatalogMTL with Negation under Stable Model Semantics]
        {The Stable Model Semantics of Datalog with Metric Temporal Operators\footnote{This is an invited submission resulting
        from an earlier KR conference publication.}}

  \author[P. A. Wałęga et al.]
         {PRZEMYSŁAW A. WAŁĘGA, DAVID J. TENA CUCALA, 
         BERNARDO CUENCA GRAU\\
         Department of Computer Science, University of Oxford, UK\\
         \email{\{przemyslaw.walega, david.tena.cucala,bernardo.cuenca.grau\}@cs.ox.ac.uk}
         \and EGOR V. KOSTYLEV\\
         Department of Informatics, University of Oslo, Norway\\
         \email{egork@ifi.uio.no}}

\jdate{March 2003}
\pubyear{2003}
\pagerange{\pageref{firstpage}--\pageref{lastpage}}
\doi{S1471068401001193}

\newtheorem{theorem}{Theorem}[section]
\newtheorem{lemma}[theorem]{Lemma}
\newtheorem{proposition}[theorem]{Proposition}
\newtheorem{claim}[theorem]{Claim}
\newtheorem{example}[theorem]{Example}
\newtheorem{definition}[theorem]{Definition}

\begin{document}

\label{firstpage}
 
\maketitle

\begin{abstract}
We introduce negation under the stable model semantics in DatalogMTL---a  temporal extension of Datalog with metric temporal operators. As a result, we obtain a rule language which combines the power of answer set programming with the temporal dimension provided by metric operators. We show that, in this setting, reasoning becomes undecidable over the rational timeline, and decidable in \EXPS{} in data complexity over the integer timeline. We  also show that, if we restrict our attention to  forward-propagating programs, reasoning over the integer timeline becomes \PS{}-complete in data complexity, and hence, no harder than over positive  programs; however, reasoning over the rational timeline in this fragment remains undecidable.
Under consideration in Theory and Practice of Logic Programming (TPLP).
\end{abstract}


\section{Introduction}

\MTL{} \cite{datalogMTL} extends positive Datalog \cite{DBLP:books/aw/AbiteboulHV95} with operators
from metric temporal logic (MTL) \cite{koymans1990specifying} interpreted
 over the rational or the integer timeline.
 For example, the following \MTL{} rule can be used
 to state that a bus driver should not work for more than six months (i.e., half a year)
 in a row:
  $$ \textit{OnLeave}(x) \gets \textit{BusDriver}(x) \wedge \boxminus_{[0,0.5]} \textit{Working}(x),$$
where the expression $\boxminus_{[0,0.5]} \textit{Working}(x)$ holds at time $t$ if $\textit{Working}(x)$ 
 holds continuously in the time interval $[t-0.5, t]$. 
 Some other examples of expressions allowed in \MTL{} are $\diamondminus_{[t_1,t_2]} \varphi$, 
which holds at time $t$ if $\varphi$ holds at some instant within the time interval $[t-t_2, t-t_1]$,
and $\boxplus_{[t_1,t_2]} \varphi$, which uses the ``future-oriented'' version of the operator $\boxminus$
and holds at time $t$ if $\varphi$ holds continuously in the time interval $[t+t_1, t+t_2]$.
 A \MTL{} dataset consists of facts involving intervals, such as $Working(\mathit{alex})@[2022,2023]$, stating that Alex was working continuously in the time interval $[2022,2023]$.
 \MTL{} is thus a very expressive language which allows a user to capture complex definitions, regulations, or events 
 involving temporal intervals.
\MTL{} is powerful enough to capture prominent temporal extensions of Datalog  such as $\mbox{Datalog}_{1S}$   
\cite{chomicki1988temporal,chomicki1989relational} 
and Templog~\cite{abadi1989temporal}, and it 
has found applications in areas such as ontology-based data access~\cite{datalogMTL},
stream reasoning \cite{WalegaAAAI},
and
similar ideas were explored in logic programming \cite{brzoska1998programming}.
Reasoning in \MTL{} is known to be \textsc{PSpace}-complete in data complexity over both
the rational \cite{DBLP:conf/ijcai/WalegaGKK19} and the integer timeline \cite{walega2020datalogmtl}.  
  
Motivated by a range of applications,
there has recently been a growing interest in 
logics that combine non-monotonic negation with
temporal constructs  
\cite{DBLP:conf/eurocast/CabalarV07,aguado2013temporal,DBLP:journals/tplp/CabalarKSS18,cabalar2020towards,DBLP:journals/ai/BeckDE18,zaniolo2012streamlog}.
Recently, we have  proposed
an extension of \MTL{} over the rationals with
stratified negation-as-failure, where rules can have negated atoms in the body,
but there can be no recursion involving predicates mentioned in such atoms.
With such an extension of \MTL{}, we can express, for example, a bus company's
policy that any vehicle older than 12 years must be decommissioned permanently
unless it has passed a special inspection in the last year:
 $$ \boxplus_{[0,\infty)} \textit{Decommis}(x) \gets  \diamondminus_{(12, \infty)}\textit{Manufactured}(x) \wedge \nott \diamondminus_{[0,1]} \textit{PassInspect}(x).$$
We also showed that
the additional expressive power provided by this type of negation does not
increase the complexity of reasoning regardless of whether we consider
the rational or the integer timeline
\cite{tena2021stratified}. The restriction to stratifiable programs, however,
 significantly limits the applicability \MTL{} to certain types of use cases.

In this paper we take a further step in this direction and
consider \MTL{} equipped with non-stratifiable  
negation interpreted under the stable model semantics \cite{DBLP:conf/iclp/GelfondL88,DBLP:journals/jar/BrooksEEMR07,DBLP:conf/asp/NogueiraBGWB01}.
This extension paves the way for 
the use of \MTL{} in 
applications where 
 derived information can be retracted in light of new evidence,
minimality of models is required,
or 
temporal inertia rules need to be formalised.
For instance, consider a bus company with the policy
that vehicles that have been serviced at a given time $t$
are automatically booked for a routine service in a year's time 
(i.e., at time $t+1$, represented by metric operator $\boxplus_{[1,1]}$), but this appointment 
must be cancelled if the bus is serviced again before then---that is, within the interval $(t + 0, t+1)$, represented by $\diamondplus_{(0,1)}$.
This policy can be written using the rule
$$
\boxplus_{[1,1]} \textit{Service}(x) \gets \textit{Service}(x) \land\, \nott \diamondplus_{(0,1)} \textit{Service}(x), 
$$
which is not stratifiable as it involves recursion via negation. 

Our setting is closely related to the recent
research on combining answer set programming (ASP)
with temporal logics.
For example, TEL \cite{DBLP:conf/eurocast/CabalarV07,aguado2013temporal,DBLP:journals/tplp/CabalarKSS18}
combines ASP with linear temporal logic, and LARS
combines ASP with window-based temporal constructs for
stream reasoning  \cite{DBLP:journals/ai/BeckDE18}.
The logic recently proposed by
\citeN{cabalar2020towards} is perhaps the closest
to our work, as it combines
stable model semantics with propositional
MTL interpreted over the natural numbers;
this logic, however,
is rather different from 
\MTL{}, where the use of logical connectives and
MTL operators is
restricted in the spirit of Datalog to disallow
disjunction and `existential' MTL operators (such as \emph{diamond},
\emph{since}, or \emph{until} operators)  in rule heads.
As we show in our paper, considering 
a language with such restrictions can lead to  favourable computational behaviour.

Our contributions in this paper are summarised as follows. 
In \Cref{negation} we present our extension (\MTLneg{}) of
\MTL{} with negation under stable model semantics.
Our language is defined  similarly to other temporal ASP formalisms; 
it extends both \MTL{}
with stratified negation and Datalog with stable model
negation.
To capture the semantics of stable models, we use interpretations 
similar to those of the \emph{here-and-there} intuitionistic  logic~\cite{heyting1930formalen,pearce1996new}. 
The main reasoning problem we consider is the existence of a stable model for a program and a dataset.
We show in  \Cref{sec::undecidability} that, in this setting, reasoning over the rational timeline is undecidable;
furthermore, undecidability holds even for propositional
forward-propagating programs (where rules cannot propagate
information towards past time points) and to data containing only bounded intervals (i.e., where
 endpoints of all intervals are rational numbers).
To regain decidability,  in \Cref{sec:integers}  we focus  on the
integer timeline.
We show in \Cref{general} that discreteness of the timeline has 
a crucial influence on the  computational behaviour, as 
reasoning becomes  decidable 
and in \EXPS{} in data complexity; this is shown by exploiting 
B\"uchi automata and their complements to find candidate stable models and verify their minimality.
Then, in \Cref{forward} we show that,  when restricted to 
forward-propagating (or, dually, to backwards-propagating)
 programs  and bounded datasets, reasoning becomes
 \PS-complete and hence no harder than for  negation-free \MTL{} \cite{DBLP:conf/ijcai/WalegaGKK19,walega2020datalogmtl}.
This is   in stark contrast with the
 undecidability of the same fragment over the rational numbers.

\section{Preliminaries}\label{preliminaries}

In this section we recapitulate the basics of temporal intervals over the integers or the rationals, and 
introduce the syntax and semantics of metric temporal operators.

\subsection{Timelines and Temporal Intervals}

A \emph{timeline} 
$\T$ is either the set $\Q$ of rationals or the set $\Z$ of integers. 
A $\T$-\emph{time point} is an element of $\T$.
A \mbox{$\T$-\emph{interval}} 
$\varrho$ is a subset of $\T$ satisfying both of the following properties:
\begin{itemize}
\item[--] $t \in\varrho$ for all $t_1,t_2 \in \varrho$ and $t \in \T$ such that 
$t_1<t<t_2$; and
\item[--]  the greatest lower bound $\varrho^-$ and the least upper bound $\varrho^+$ of $\varrho$ are both in $\T\, \cup \{-\infty ,\infty\}$.
\end{itemize}
The bounds $\varrho^-$ and $\varrho^+$ are called the \emph{left} and \emph{right endpoints} of $\varrho$, respectively.
A $\T$-interval is  \emph{punctual} if it contains exactly one $\T$-time point, it is
\emph{non-negative} if it contains no negative $\T$-time points, it is
\emph{bounded} if  both its endpoints are in  $\T$, and it is \emph{closed}
if it includes both of its endpoints.
In these and similar notions, we
often omit the reference to $\T$ if it is
clear from the context.  
We consider binary representations of integers and
fractional representations of rationals, with an integer 
numerator and a positive integer denominator, encoded in binary.
We use standard representations of the form 
$\langle \varrho^-,\varrho^+ \rangle$ for a non-empty 
interval
$\varrho$ (i.e., $\varrho \cap \T \neq \emptyset$),  
where the \emph{left bracket}
$\langle$  is either  $[$ or $($, the \emph{right bracket} $\rangle$ is
either $]$ or $)$, and, if numeric, the endpoints $\varrho^-$ and $\varrho^{+}$ are represented as explained above.
Brackets $[$  and $]$ indicate that the
 endpoints are included in the interval, whereas $($ and 
$)$ indicate that they are not included; note that, by this convention, $[$ and $]$
cannot be used with endpoints $-\infty$ and $\infty$.
We often abbreviate a punctual interval $[t,t]$ as $t$.
If it is clear from the context, we abuse
notation and identify each interval representation with the 
interval it represents.

\subsection{Syntax and Semantics of Metric Temporal Expressions}

Assume a function-free first-order vocabulary and a
timeline $\T$.
A \emph{relational atom} is an expression of the form $P(\sbf)$, 
where 
$P$ is a predicate  and $\sbf$ is a tuple of constants and variables
of the same arity as $P$.
A  \emph{metric atom} is an expression given by the following
grammar, where $P(\sbf)$ ranges over relational atoms and $\varrho$ over non-empty, non-negative intervals:
\begin{equation*}
	M  \Coloneqq  \top \mid \bot \mid 
	P(\sbf) \mid 
	\diamondminus_{ \varrho} M \mid 
	\diamondplus_{\varrho} M \mid  
	\boxminus_{\varrho}M \mid 
	\boxplus_{\varrho} M \mid  
	M \So_{\varrho} M \mid  
	M \Uo_\varrho M.
\end{equation*}
A metric atom is \emph{ground} if it mentions no variables.
A \mbox{\emph{metric fact}} is an expression $\matA @ \varrho$, with
$\matA$ a ground \mbox{metric} atom 
and $\varrho$ a non-empty $\T$-interval; it is \emph{relational} if so is $\matA$.
A~\emph{dataset} is a finite set of relational facts;
 it is \emph{bounded} if so are all intervals it mentions.
For a dataset $\D$, we denote with
$t^{\min}_{\D}$ and $t^{\max}_{\D}$ 
the smallest and the largest numbers mentioned in $\D$; if $\D$ mentions 
no numbers, we let
$t^{\min}_{\D}=t^{\max}_{\D}=0$.

An \emph{interpretation} $\I$ is a function which assigns to each time point $t \in \T$
a set of ground relational atoms;
if  an atom $P(\cbf)$ belongs to this set,  we say that $P(\cbf)$  is \emph{satisfied} at $t$ in $\I$ and we write $\I,t \models_\T P(\cbf)$.
This notion extends to other ground metric atoms 
as given in \Cref{tab::semantics}.
Interpretation $\I$ is a \emph{model} of a metric fact $M @ \varrho$, 
written $\I \models_\T M @ \varrho$, if $\I,t \models M$ for all
$t \in \varrho$; 
it is a  \emph{model} of a set $\M$ of metric facts (e.g., a dataset)
if it is a model of all facts in $\M$. 
Set $\M$  \emph{entails} a set $\M'$ of metric facts, written $\M \models \M'$, if every model of $\M$ is a model of $\M'$.
An interpretation $\I$ \emph{contains} an
interpretation~$\I'$, written  $\I' \subseteq \I$, if for 
each ground relational atom
$P(\cbf)$ and each time point $t \in \T$, having $\I',t \models_\T P(\cbf)$ implies $\I,t \models_\T P(\cbf)$. 
Furthermore, $\I$ is the \emph{least} interpretation in a set $X$ of interpretations, if $\I \subseteq \I'$ for every $\I' \in X$.

\begin{table}[t]
	\begin{alignat*}{3}
		&\I,{ t} \models_\T  { \top}    && && \text{for each } t \in \T
		\\
		&\I,{ t}  \models_\T { \bot}   && && \text{for no } t \in \T
		\\
		&\I,{ t}  \models_\T  {\diamondminus_\varrho} M    && \text{iff}   && \I,{t'} \models_\T  { M} \text{ for some } t' \text{ with } t -t' \in \varrho
		\\
		&\I,{  t}  \models_\T { \diamondplus_\varrho}  M  && \text{iff} &&  \I,{t'} \models_\T  { M} \text{ for some } t'  \text{ with } t' - t \in \varrho
		\\
		&\I,{ t}  \models_\T  {\boxminus_\varrho}  M  && \text{iff} && \I,{ t'} \models_\T  { M} \text{ for all } t' \text{ with } t - t' \in \varrho 
		\\
		&\I,{ t}  \models_\T  {\boxplus_\varrho}  M   && \text{iff} && \I,{ t'}\models_\T  { M} \text{ for all } t'  \text{ with } t' - t \in \varrho 
		\\
		&\I,{ t}  \models_\T  { M_1}  \So_\varrho M_2 \quad   &&  \text{iff} \quad  &&  \I,{ t'} \models_\T  { M_2} \text{ for some } t'  \text{ with } t - t' \in \varrho
		\text{ and }  \I,{t''} \models_\T  M_1 \text{ for all } t'' \in (t',t) 
		\\
		&\I,{ t}  \models_\T  { M_1}  \Uo_\varrho M_2   &&  \text{iff}  && \I,{ t'} \models_\T  { M_2} \text{ for some } t'  \text{ with } t' - t \in \varrho  
		\text{ and }  \I,{ t''} \models_\T  M_1 \text{ for all } t'' \in (t,t') 
	\end{alignat*}
	\caption{Semantics of ground metric atoms}
	\label{tab::semantics}
\end{table}

\section{DatalogMTL with Negation Under Stable Model Semantics}\label{negation}

In this section we propose \MTLneg, which
extends \MTL{} with
stratified negation as defined by  \citeN{tena2021stratified},
to support unstratified
use of negation in rules interpreted
under stable model semantics. 

The syntax of \MTLneg{}  is the natural extension
of the positive case:
rule bodies are conjunctions of atoms and
negated atoms, whereas rule heads are  defined
as in  \MTL{}. 
Forward-propagating \MTLneg{}  
 programs are also defined analogously to the positive case
 \cite{WalegaAAAI}, by requiring that rule  bodies and heads
 do not mention metric operators
referring to the future and to the past, respectively.
  
\begin{definition} \label{rule}
A \emph{rule} $r$ is an expression of the form
\begin{equation}
M \gets M_1 \wedge \dots \wedge M_k \wedge \nott M_{k+1} \wedge \dots \wedge \nott  M_{m} , \qquad (m \geq k \geq 0)
\label{eq:ruleneg}
\end{equation}
where each $M_i$ is a metric atom, and $M$ is a metric atom
specified by the following grammar, where $P(\sbf)$ ranges over relational atoms and $\varrho$ over non-empty  non-negative intervals:
\begin{equation*} 
M \Coloneqq  
\top \mid \bot \mid P(\sbf) \mid 
\boxminus_\varrho M \mid 
\boxplus_\varrho M . 
\end{equation*} 
The \emph{head} of $r$ is the consequent $M$ and the \emph{body}
is the conjunction in the antecedent, where 
${M_1, \ldots, M_k}$ are its \emph{positive} body atoms, and ${M_{k+1}, \ldots, M_{m}}$ are its \emph{negated} body atoms.
A rule is \emph{safe} if each variable it mentions in the head occurs in some positive body atom in a position other than
a left operand of $\So$ or $\Uo$.
A rule is \emph{ground} if it has no variables, and it is \emph{positive} if it has no negated body atoms. 
A  $(\MTLneg{})$ \emph{program} is a finite set  of safe  rules; it is  \emph{ground} or \emph{positive} if all its rules are. For a program $\Prog$, we let $\mathsf{ground}({\Prog})$ 
be the set of all ground rules that can be obtained by replacing variables in $\Prog$ with constants.
A program is  \emph{forward-propagating} (\MTLfp{}) if it is \MTLneg{} but does not mention the operators $\diamondplus$, $\boxplus$, and $\Uo$  in rule bodies, or the 
operator $\boxminus$ is rule heads.
\end{definition}

The definition of stable models for  Datalog with negation relies
on  the reduct construction by \citeN{DBLP:conf/iclp/GelfondL88}, 
which has been adapted to various extensions of ASP 
 \cite{faber2004recursive}.
Such reduct constructions, however,
do not have a natural 
equivalent in \MTLneg, where
interpretations may satisfy a fact  at some, but not all points of 
the infinite timeline, and it is thus
unclear which rules or atoms should be included in the reduct.

Following the approach of
\citeN{DBLP:conf/eurocast/CabalarV07} and \citeN{cabalar2020towards},
we define stable models
for \MTLneg~analogously to the models of  
\emph{equilibrium logic} \cite{pearce1996new}, which
in turn are defined in terms
of interpretations for the \emph{here-and-there} intuitionistic logic~\cite{heyting1930formalen}.
In this logic, each interpretation is an ordered pair $(H,T)$ of sets 
$H$ (``here'') and $T$ (``there'') of relational propositional (i.e.\ using only predicates of arity $0$) atoms such that $H \subseteq T$.
We therefore start by generalising such interpretations
to the context of \MTLneg.
For the remainder of this section, we fix a timeline $\T$,
which will be implicit in all our definitions and technical results.

\begin{definition} \label{def::HT}
An \emph{HT-interpretation} is  a pair $( \Hmod, \Tmod )$ of \mbox{interpretations} such that $\Hmod \subseteq \Tmod$. 
An HT-interpretation  $( \Hmod, \Tmod )$ is an
\emph{HT-model} of a dataset $\D$ if $\Hmod$ is a model of $\D$; furthermore,
it is an \emph{HT-model} of a
rule~$r$ if,
for each rule of  Form~\eqref{eq:ruleneg} in $\groundp{\Prog}$ and for each $t\in\T$,
 the following  hold:
\begin{enumerate}[leftmargin=.3in]
\item If  $\Hmod ,t \models_\T M_i$ for all $i \in \{1, \ldots, k\}$ and ${\Tmod ,t \not\models_\T  M_j}$ for all $j \in \{k+1, \ldots, m\}$, then ${\Hmod ,t \models_\T  M}$. 

\item If  $\Tmod ,t \models_\T  M_i$ for all $i \in \{1,\ldots, k\}$ and ${\Tmod ,t \not\models_\T M_j}$ for all $j \in \{k+1, \ldots, m\}$, then ${\Tmod ,t \models_\T  M}$.
\end{enumerate}
Finally, $( \Hmod, \Tmod )$  is an  \emph{HT-model} of a program if it is an HT-model of all its rules.
\end{definition}
An HT-interpretation is, therefore, a pair  $( \Hmod, \Tmod )$ of 
standard interpretations. Interpretation $\Hmod$
is  contained
in $\Tmod$ and determines whether a dataset is satisfied.
Although both interpretations are used to evaluate rules, it is 
$\Tmod$
which evaluates negated body atoms and thus determines
when a rule with negation can be `triggered'. 
When this happens, the rule ensures that
if any of $\Hmod$ or $\Tmod$ satisfies the positive body atoms, then
it will also satisfy the head. 
Since $\Hmod \subseteq \Tmod$ in an HT-interpretation, all relational facts entailed by $\Hmod$ are also entailed by $\Tmod$.
We show next that this property can be generalised to arbitrary metric facts.

\begin{proposition} \label{prop::HT}
For every HT-interpretation  $( \Hmod, \Tmod )$,  metric atom $M$, and time point $t$, 	
if $\Hmod ,t \models  M$ then $\Tmod ,t \models M$.
\end{proposition}
\begin{proof}
We proceed by induction on the structure of $M$.
If $M$ is $\top$ or $\bot$, the claim holds trivially, and if
$M$ is a relational atom, then the claim holds by $\Hmod \subseteq \Tmod$.
For the inductive step it suffices to consider the cases when $M$ 
 is of the form $\boxminus_{\varrho} M_1$ or $M_1
\So_{\varrho} \matA_2$, for some interval $\varrho$ and metric atoms $M_1$ and $M_2$.
Indeed, if $\matA$ is of the form $\diamondminus_{\varrho} M_1$ or $\diamondplus_{\varrho} M_1$, then it is equivalent to $\top
\So_{\varrho} M_1$ or $\top \Uo_{\varrho}
M_1$, respectively, while the cases when $M$  is of the form $\boxplus_{\varrho} M_1$ or $M_1 \Uo_{\varrho} M_2$ are symmetric to the cases of 
$\boxminus_{\varrho} M_1$ or  $M_1
\So_{\varrho} M_2$, respectively.

If $M$ is $\boxminus_{\varrho } M_1$, then $\Hmod, t \models M$ implies that  $\Hmod, t' \models M_1$, for all $t'$ such that $t-t'\in\varrho$.
Hence, by the inductive hypothesis, $\Tmod, t' \models M_1$  for all $t'$ such that $t-t'\in\varrho$, and so, $\Tmod, t \models \boxminus_{\varrho } M_1$.
Similarly, if $M$ is $M_1
\So_{\varrho} M_2$, then 
there is $t'$ with $t-t'\in\varrho$ such that  $\Hmod, t' \models M_2$ and $\Hmod, t'' \models M_1$, for all $t'' \in (t',t)$.
By the inductive hypothesis we obtain that  $\Tmod, t' \models M_2$
and $\Tmod, t'' \models M_1$, for all $t'' \in (t',t)$, so $\Tmod, t \models M_1
\So_{\varrho} M_2$.
\end{proof}

Although the converse statement does not always hold,
we can nonetheless prove the following result, which will underpin
 our definition of stable models.

\begin{theorem} \label{thm::least}
Let $(\Tmod,\Tmod)$ be an HT-model of a program $\Prog$ and a dataset $\D$.
Then the set of interpretations $\{  \Hmod \mid (\Hmod,\Tmod) \text{ is an HT-model of  $\Prog$ and $\D$} \}$ contains a unique least interpretation.
\end{theorem}
\begin{proof}
We use transfinite induction to construct a sequence
of interpretations $\Hmod_0$, $\Hmod_1$, $\dots$,
where each interpretation is contained in $\Tmod$.
We will then show that $\Hmod_{\omega_1}$, where $\omega_1$ is the first uncountable
ordinal, is the least amongst all interpretations $\Hmod$
such that $(\Hmod,\Tmod)$ is an HT-model of $\Prog$ and $\D$.

Let $\Hmod_0$ be the least model of $\D$.
Then, for each successor ordinal $\alpha$, let 
$\Hmod_{\alpha}$ be the least interpretation
satisfying the following property:
for each rule of Form~\eqref{eq:ruleneg} in $\groundp{\Prog}$, and for each time point $t$,   
if ${\Hmod_{\alpha-1},t \models \matA_i}$ for each $1 \leq i \leq k$ and 
${\Tmod,t \not\models \matA_j}$ for each $k+1 \leq j \leq m$,
then ${\Hmod_{\alpha},t \models M}$.
Moreover, for each limit ordinal $\alpha$, we define 
$\Hmod_{\alpha}$ as  $\bigcup_{\beta < \alpha} 
\Hmod_\beta$.
By induction on ordinals $\alpha$ we can show simultaneously that 
$\Hmod_{\alpha}$ is well-defined and that $\Hmod_\alpha \subseteq \Tmod$.
For the basis of the induction, we recall that $\Hmod_0$ is the least model of $\D$, so $\Hmod_0 $ is well-defined. Moreover, since  $(\Tmod,\Tmod)$ is an HT-model 
of $\D$, we obtain that $\Hmod_0 \subseteq \Tmod$.
Now, consider the inductive step for a successor ordinal $\alpha$. 
To show that $\Hmod_{\alpha}$ is well-defined it suffices to show that for each rule of Form~\eqref{eq:ruleneg} in
$\groundp{\Prog}$, and for each time point $t$,   
if ${\Hmod_{\alpha-1},t \models \matA_i}$ for each $1 \leq i \leq k$ and 
${\Tmod,t \not\models \matA_j}$ for each $k+1 \leq j \leq m$,
then $M$ is not $\bot$.
Indeed, by the inductive assumption we have $\Hmod_{\alpha-1} \subseteq \Tmod$, so if $M$ is $\bot$, then 
 $\Tmod,t \models \bot$, which contradicts the assumption that  $(\Tmod,\Tmod)$ is an
HT-model of $\Prog$.
Moreover, since $\Hmod_{\alpha-1} \subseteq \Tmod$ and $(\Tmod,\Tmod)$ is an
HT-model of $\Prog$, we need to have $\Tmod,t \models M$, so $\Hmod_\alpha \subseteq \Tmod$.
The inductive step for a limit ordinal $\alpha$ holds trivially, since $\Hmod_{\alpha}$ is defined as  $\bigcup_{\beta < \alpha} 
\Hmod_\beta$.

We thus obtain that $\Hmod_{\omega_1} \subseteq \Tmod$, and so, $(\Hmod_{\omega_1},\Tmod)$ is 
an HT-interpretation. By construction, $\Hmod_{\omega_1}$ contains $\Hmod_0$ and therefore
$(\Hmod_{\omega_1},\Tmod)$ is an HT-model of $\D$.
It is also an HT-model of $\Prog$, since $\omega_1$ rounds of rule applications are enough to ensure that
$\Hmod_{\omega_1}$ is a fixpoint with respect to the application of the rules of $\Prog$
 \cite{brandt2017ontology}.
Finally, to show that $\Hmod_{\omega_1}$ is the least among interpretations $\Hmod$ such that $(\Hmod,\Tmod)$ is an HT-model of $\Prog$ and $\D$, consider any such $\Hmod$.
Using transfinite induction in a way similar to the previous paragraph, one can show that $\Hmod_{\alpha} \subseteq \Hmod$
for each ordinal $\alpha$, and thus $\Hmod_{\omega_1} \subseteq \Hmod$.
\end{proof}

Given a program $\Prog$, a dataset $\D$, and  
an interpretation $\Tmod$ such that 
$(\Tmod,\Tmod)$ is an HT-model of $\Prog$, we let $\Hmod^\Tmod_{\Prog,\D}$ denote
the least interpretation whose existence
is guaranteed by \Cref{thm::least}.

In equilibrium logic, a model of a program
 is a set $T$ of relational propositional atoms that satisfies
 the rules of the program and for which there exists no set $H \subsetneq T$ 
such that $(H,T)$ is a model of the program in here-and-there logic.
This ensures that equilibrium logic implements a kind of minimal model
reasoning.
We next generalise this notion to \MTLneg{} by building on our previous definition of
the least interpretation $\Hmod^\Tmod_{\Prog,\D}$.

\begin{definition} \label{def::stable}
An interpretation $\Tmod$ is a \emph{stable model} of 
a program $\Prog$ and a dataset $\D$
if and only if
$(\Tmod,\Tmod)$ is an HT-model of $\Prog$ and $\D$, and 
${\Hmod^\Tmod_{\Prog,\D} = \Tmod}$.
\end{definition}
 
\begin{example}\label{example_model}
Consider a dataset with a single fact $P@[0,1]$ and a propositional 
\MTLneg{} program consisting of a single rule $R \gets \diamondminus_1 P \land \nott Q$.
In this setting, there is just a single
stable model,  namely the interpretation where $P$ holds at all time points from $[0,1]$, $R$ holds at all time points from $[1,2]$, and no relational atoms are satisfied anywhere else.

Next, consider a dataset with facts $P@0$ and $Q@1$, together with a propositional 
\MTLneg{} program that consists of two  rules $R \gets  P \land \nott \diamondplus_1 R$ and
$R \gets  Q \land \nott \diamondminus_1 R$.
This dataset and program have two stable models.
In the first model, $P$ and $R$ hold at 0, whereas $Q$ holds at 1.
In the second model, $P$ holds at 0, whereas $Q$ and $R$ hold at 1.

Finally, let us consider the empty dataset and a program consisting
of rules $Q \gets \nott P$ and ${P \gets \nott Q}$. Syntactically, 
this is not only a propositional \MTLneg{} program, but also
a standard answer set program (ASP).  According to our temporal semantics,
this program and dataset admit infinitely many stable models:
for each set  $X$ of time points, there is a stable model in which $P$ holds at each time point in
$X$ and $Q$ holds at all other time points.
In contrast, the same program under the standard ASP semantics has only two stable models, namely $\{ P \}$ and $ \{ Q \}$.
\end{example}
 
We next show that our semantics for \MTLneg{}  also generalises the 
semantics of (positive) \MTL{}  programs.
If a \MTL{} program $\Prog$ and a dataset $\D$ have a model, they also admit
a least model \cite{brandt2017ontology}.
This can be equivalently reformulated by stating that
if the set of all interpretations $\{ \I \mid (\I,\I) \text{ is an HT-model of $\Prog$ and $\D$} \}$ is not empty, then this set contains a unique least interpretation.

\begin{theorem}\label{thm::positive}
Let $\Prog$ be a positive program and let $\D$ be a dataset. 
An interpretation $\I$ is a stable model of $\Prog$ and $\D$ if and only if $\I$ is their least model.
\end{theorem}
\begin{proof}
We can first use the fact that
$\Prog$ is positive to show that 
if $(\I,\I)$ is an HT-model
of $\Prog$ and $\D$, then
$\Hmod^\I_{\Prog,\D}$ is the least model of $\Prog$ and $\D$.
Indeed, if $(\I,\I)$ is an HT-model
of $\Prog$ and $\D$,
we can define the sequence
$\Hmod_0$, $\Hmod_1$, $\dots$ of interpretations contained in
$\I$ as in the proof of \Cref{thm::least}, which satisfies $\Hmod_{\omega_1} = \Hmod^\I_{\Prog,\D} $.
Furthermore, since $\Prog$ is positive,  we can observe that, for every ordinal $\alpha$, it holds that $ { \Hmod_\alpha = T_\Prog^\alpha(\I_\D) }$,
where $T_\Prog^\alpha(\I_\D)$ is the result of applying $\alpha$ times the immediate consequence operator of a positive program $\Prog$ to an interpretation $\I_\D$ represented by $\D$. In particular, $ { \Hmod_{\omega_1} = T_\Prog^{\omega_1}(\I_\D) }$,
which is the least model of $\Prog$ and $\D$ \cite{walkega2021finitely,brandt2017ontology}.
Hence, $\Hmod^\I_{\Prog,\D}$ is the least model of $\Prog$ and $\D$.

Now, if $\I$ is a stable model of $\Prog$ and $\D$, then $(\I,\I)$ is an HT-model
of $\Prog$ and $\D$; as shown in the previous paragraph, 
this implies that $\Hmod^\I_{\Prog,\D}$ is the least model of
$\Prog$ and $\D$. However, since
$\I$ is a stable model of $\Prog$ and $\D$, we have
$\I = \Hmod^\I_{\Prog,\D}$, and thus $\I$ is also the least model of $\Prog$ and $\D$.
Conversely, if $\I$ is the least model of $\Prog$ and $\D$, then 
 $(\I,\I)$ is an HT-model of $\Prog$ and $\D$; then, as shown in the previous
 paragraph, $\I = \Hmod^\I_{\Prog,\D}$, and so $\I$ is a stable model
 of $\Prog$ and $\D$. 
\end{proof}
It follows that,  if a positive program and a dataset have a model,
 then they have a stable model. 
Note, however, that this is not the case
for  other  temporal logics with stable model semantics
\cite{cabalar2011automata,bozzelli2015complexity},
and the reason why this property holds in our setting is given by  \Cref{thm::least}.

Finally, our semantics also
generalises  
that of  stratifiable \MTLneg{} programs~\cite{tena2021stratified}, where
rules do not have cyclic dependencies via negation and can be
organised in strata. 
Each such a stratifiable, \mbox{$\bot$-free} program $\Prog$ and dataset $\D$ have a single
 least model constructed by  applying to $\D$ rules of $\Prog$ stratum by stratum in a minimal way \cite{tena2021stratified}.
As in the case of positive programs, 
we can show that such least model corresponds to the single stable model of $\Prog$ and $\D$.
Hence, analogously to the case of plain Datalog, positive and stratifiable 
\MTL{} programs cannot have multiple stable models.
Arbitrary programs, however, can have any  number of stable models, which is witnessed by \Cref{example_model}.

In the rest of the paper we study decidability and complexity of 
\emph{reasoning}, which is (in the context of this paper) the problem of
checking if a \MTLneg{}
program $\Prog$ and a dataset $\D$ have a stable model.
We focus on \emph{data
complexity}---that is, we assume that $\Prog$ is fixed and
only $\D$ forms the input---which is the most relevant measure if complexity in
data intensive applications.

\bigskip

Before we close this section, however, it is worth pointing out the
connections between the problem of checking existence
of a stable model and the related problem of \emph{fact entailment}, as defined next.
Following the standard conventions of
 non-monotonic logics and answer set programming \cite{eiter2009answer}, we say that a \MTLneg{} program $\Prog$ and a dataset $\D$
\emph{bravely} 
entail a relational fact $P(\cbf)@\varrho$ if 
$\I \models_\T  P(\cbf)@\varrho$ for some stable model $\I$ of $\Prog$ and $\D$.
and we say that $\Prog$ and $\D$
\emph{cautiously}
entail  $P(\cbf)@\varrho$ if 
$\I \models_\T  P(\cbf)@\varrho$ for all stable models $\I$ of $\Prog$ and $\D$.
The problem of brave (resp.\ cautious) fact entailment consists in deciding whether a \MTLneg{} program
and a dataset bravely (resp.\ cautiously) entail a given relational fact.
As we show next, both variants of the problem are inter-reducible
with checking existence (or non-existence) of a stable model.
Moreover, we argue that these reductions allow us to transfer bounds for
\emph{data complexity} which, for fact entailment, refers to the setting where both the program $\Prog$ and the
fact $P(\cbf)@\varrho$ are fixed, and only the dataset $\D$ constitutes the input.

\begin{proposition}
In \MTLneg{}, existence of a stable model can
 be reduced in \AC{} to (i) brave fact entailment, 
and to (ii) the complement of cautious fact entailment, and vice versa. Furthermore, 
the reductions involved do not depend on the input dataset.
\end{proposition}
\begin{proof}
We start by showing Statement (i).
To check if a  \MTLneg{} program $\Prog$ and a dataset $\D$ have a stable model,
it suffices to
add to $\D$ a fact $P@0$ with a fresh proposition $P$ and check   whether $\Prog$ and the extended dataset  bravely entail $P@0$.
To check if $\Prog$ and $\D$ bravely entail a relation fact $P(\cbf) @ \varrho$,
it suffices to verify that the following program~$\Prog'$ and dataset $\D'$ have a stable model:
\begin{align*}
\Prog' = \Prog \cup \{
\bot \gets P'(\xbf) \land \nott \boxminus_{\varrho_1} P(\xbf),  
\bot \gets P'(\xbf) \land \nott \boxplus_{\varrho_2} P(\xbf) 
\},
&&
\D'   =  \D  \cup \{P'(\cbf) @ t \},
\end{align*}
where $P'$ is a fresh predicate
of the same arity as $P$, $\xbf$ is a tuple of distinct variables,
$t$ is an arbitrary time point belonging to $\varrho$, 
whereas $\varrho_1$ and $\varrho_2$ depend on both $\varrho$ and $t$;
for example, 
if $\varrho =[ t_1,t_2 )$, then
$\varrho_1=[0,t-t_1]$ and $\varrho_2=[0,t_2-t)$, where if $t_2= \infty$, then $t_2-t$ stands for $\infty$.

Next, we show Statement (ii).
To check if   $\Prog$ and $\D$ have a stable model, it  
suffices to check if they do not cautiously entail a fact $P@0$, where $P$ is a fresh proposition \cite[Proposition~3]{datalogMTL}.
On the other hand, to check if $\Prog$ and $\D$ 
do not cautiously entail  a relational fact 
$P(\cbf) @ \varrho$, it suffices to verify that the following program~$\Prog''$ and dataset $\D''$
 have a stable model:
\begin{align*}
\Prog'' = \Prog \cup \{\bot \gets P'(\xbf) \land \boxminus_{\varrho_1} P(\xbf) \land \boxplus_{\varrho_2} P(\xbf)  \},  
&&
\D''   =  \D  \cup \{P'(\cbf) @ t \},
\end{align*}
where $P'$, $\xbf$, $t$, $\varrho_1$, and $\varrho_2$ are as in the proof of Statement (i).

Finally, we observe that all  the above reductions can be performed in \AC{}. 
Moreover, they allow us to 
transfer data complexity bounds, since the programs and facts we construct in the reductions do not depend on input datasets.
\end{proof}

\section{Undecidability over the Rational Timeline}
\label{sec::undecidability}

In this section we focus on the rational timeline, so let us fix
$\T = \Q$. 
In this setting, standard reasoning problems
are \PS{}-complete 
in data complexity for positive programs~\cite{DBLP:conf/ijcai/WalegaGKK19},
as well as for programs with negation, if they are stratified~\cite{tena2021stratified}.

We next show that, in stark contrast with the positive case,
reasoning in  \MTLneg{}  is undecidable.
Furthermore, undecidability holds even for programs that are propositional, forward-propagating and considered to be fixed,
and 
even if the input datasets are  bounded.

\begin{theorem} \label{thm::undec} 
Checking whether a propositional
\MTLfp{}  and a bounded dataset
have a stable model over the rational timeline
is undecidable with respect to data complexity. 
\end{theorem}
\begin{proof}
Let $\TM = (\Sigma,\states,\delta,\qinit,\qhalt)$ be a deterministic Turing machine with
 $\Sigma$ the input alphabet,
$\states$ the set of states, 
$ {\delta: \Sigma_\sqcup \times ( \states \setminus \{ \qhalt\} ) \longrightarrow \Sigma_\sqcup \times \states \times \{\Left,\Right \} } $
the transition function for  $\Sigma_\sqcup = \Sigma \cup \{\sqcup\}$ and blank symbol $\sqcup$, 
$\Left$ and $\Right$ the symbols indicating the head movements, and $\qinit,\qhalt \in \states$ 
the initial  and halting states. Without loss of generality, we assume that $\TM$  never tries to move to the left of the left-most cell of the tape.	

We construct a propositional \MTLfp{} program $\Prog_\TM$ and then 
reduce (in $\AC$) every input word $\word$ to a dataset $\D_w$ with only bounded intervals so that 
$\TM$ halts on $\word$ if and only if $\Prog_\TM$  and $\D_w$ do not have a stable model.

We represent, for each $i \geq 1$, the $i$th configuration in the computation of $\TM$ on input $w$ within
the interval $[2i,2i+2)$, as illustrated in \Cref{fig::conf}, where we assume that in the configuration the state is $q$, the head is over the $j$th cell, and the contents of the first $|w| + i$ cells of the 
tape are symbols $s_1, \dots, s_{|w|+i}$ (in the $i$th configuration,
only the first $|\word|+i$ cells can be non-empty).
The state is encoded within the first half $[2i, 2i+1]$ of the interval: 
a proposition $S_{q}$
holds at some time point within $[2i,2i+1]$.
The contents of the tape and the head position are encoded within the second half $(2i+1,2i+2)$ of the interval; in particular, $|\word|+i$ time points ${t^i_1 < \cdots < t^i_{|\word|+i}}$ in $(2i+1,2i+2)$ are used so that, for each $j \in\{1, \dots, |w| +i\}$, proposition $C_{s_j}$ holds at $t^i_{j}$, encoding the fact that $s_j$ are the contents of the $j$th cell in the configuration, and proposition $H$ holds at $t^i_j$, encoding the fact that the head is over the $j$th cell in the configuration.
We also use additional propositions: $S$, which holds all through $[2i,2i+1]$ and ensures that these intervals are only used to encode states; $N$, which holds at a single new time point in $(2i+1,2i+2)$ beyond $t^i_{|\word|+i}$ and will allow us to increase the number of time points encoding cells; $\overline{N}$ and $\overline{H}$, which simulate negations of $N$ and $H$, respectively; $\overline{C}$, which marks points not used to encode the tape contents, and $L$, used for encoding left-moving transitions.

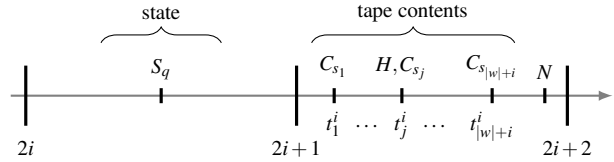
\begin{figure}[ht]
\centering
\begin{tikzpicture}
\footnotesize

\tikzset{>=latex}

\draw[->, thick, gray] (-0.2,0)--(7.8,0);

\node(t1) at (0,0){ }; 
\node(t2) at (3.6,0){ }; 
\node(t3) at (7.2,0){ }; 
\node(t4) at (1.8,0){ };
\node(t5) at (4.1,0){ }; 
\node(t6) at (6.2,0){ }; 
\node(t7) at (6.9,0){ };
\node(t8) at (5,0){ };

\node(t1_up)  at ($(t1)+(0,0.4)$) {};
\node(t1_down)  at ($(t1)+(0,-0.4)$) {};
\draw[-, very thick] (  t1 |- t1_up) -- (t1 |- t1_down);
\node[below  = 0.4  of t1] {$2i$};

\node(t2_up)  at ($(t2)+(0,0.4)$) {};
\node(t2_down)  at ($(t2)+(0,-0.4)$) {};
\draw[-, very thick] (  t2 |- t2_up) -- (t2 |- t2_down);
\node[below  = 0.4 of t2] {$2i+1$};

\node(t3_up)  at ($(t3)+(0,0.4)$) {};
\node(t3_down)  at ($(t3)+(0,-0.4)$) {};
\draw[-, very thick] (  t3 |- t3_up) -- (t3 |- t3_down);
\node[below  = 0.4 of t3] {$2i+2$};

\node(t4_up)  at ($(t4)+(0,0.1)$) {};
\node(t4_down)  at ($(t4)+(0,-0.1)$) {};
\draw[-, very thick] (  t4 |- t4_up) -- (t4 |- t4_down);
\node[above  = 0.01 of t4] {$S_q$};

\node(t5_up)  at ($(t5)+(0,0.1)$) {};
\node(t5_down)  at ($(t5)+(0,-0.1)$) {};
\draw[-, very thick] (  t5 |- t5_up) -- (t5 |- t5_down);
\node[below  = 0.01 of t5] {$t_1^i$};
\node[above  = 0.05 of t5] {$C_{s_1}$};

\node(t6_up)  at ($(t6)+(0,0.1)$) {};
\node(t6_down)  at ($(t6)+(0,-0.1)$) {};
\draw[-, very thick] (  t6 |- t6_up) -- (t6 |- t6_down);
\node[below  = 0.01 of t6] {$t_{|\word|+i}^i$};
\node[above  = 0.01 of t6] {$C_{s_{|\word|+i}}$};

\node(t7_up)  at ($(t7)+(0,0.1)$) {};
\node(t7_down)  at ($(t7)+(0,-0.1)$) {};
\draw[-, very thick] (  t7 |- t7_up) -- (t7 |- t7_down);
\node[below  = 0.01 of t7] {};
\node[above  = 0.01 of t7] {$N$};

\node(t8_up)  at ($(t5)+(0,0.1)$) {};
\node(t8_down)  at ($(t5)+(0,-0.1)$) {};
\draw[-, very thick] (  t8 |- t8_up) -- (t8 |- t8_down);
\node[below  = 0.01 of t8] {$t_j^i$};
\node[above  = 0.01 of t8] {$H, C_{s_j}$};

\node   at ($(t5)+(0.45,-0.4)$) {\dots};

\node   at ($(t5)+(1.35,-0.4)$) {\dots};

\draw [decorate,decoration={brace,amplitude=6pt,raise=20pt}]
($(t1)+(1,0)$) -- ($(t2)+(-1,0)$) node [black,midway,xshift=0cm,yshift=1.1cm] 
{state};

\draw [decorate,decoration={brace,amplitude=6pt,raise=20pt}]
($(t5)+(-0.3,0)$) -- ($(t6)+(0.3,0)$) node [black,midway,xshift=0cm,yshift=1.1cm] 
{tape contents};

\end{tikzpicture}
\caption{Encoding of the $i$th configuration}
\label{fig::conf}
\end{figure}

The first block of rules in $\Prog_\TM$ consists of the following rules, for each
 $X \in \{N,H \}$ and $s \in \Sigma_\sqcup$:
\begin{align*}
\overline{X} & \gets \nott X,
&
X & \gets \nott \overline{X},
&
\bot &\gets X \land \overline{X},
&
\bot & \gets X \land \diamondminus_{(0,1)} X,
\\
\overline{X} & \gets S,
&
\overline{N} & \gets C_{s},
&
\overline{H} & \gets \overline{C},
&
\bot & \gets S \wedge (\overline{C} \wedge \overline{N}) \So_{(0,1)}  C_{s}.
\end{align*}
The first three rules state that, at each  time point,
either $X$ or $\overline{X}$ holds. 
The fourth rule states that $X$ cannot hold twice in an open interval of length 1.
By the fifth rule, 
$X$ and $S$ cannot hold at the same time point.
The sixth rule states that $\overline{N}$ holds in all time points encoding cells.
The second to last rule states that 
$H$ does not hold in time points that do not encode cells.
The last rule ensures that after time point $t^i_{|\word|+i}$ encoding the
last relevant cell in the $i$th configuration,
there is another time point within $(2i+1,2i+2)$ where
 $N$ holds. Note that the last rule uses conjunction within 
 a metric operator, 
which is not syntactically allowed, but 
can be easily simulated by replacing $\overline{C} \land \overline{N}$ with a fresh proposition $P$ and adding a rule $P \gets \overline{C} \land \overline{N}$;
this abbreviation will be useful for simplifying other formulas used later on in the reduction.
 
The second block consists of the following rules, propagating propositions to the interval encoding the consequent configuration, for every $s \in \Sigma_\sqcup$:
\begin{align*}
\boxplus_2 S &\gets S,
&
\boxplus_2 \overline{C} &\gets \overline{C} \land \overline{N},
&
\boxplus_2 C_{\sqcup} &\gets N \land \diamondminus_{(0,\infty)}S_{\qinit},
&
\boxplus_2 C_{s} &\gets	C_{s} \land \overline{H} .
\end{align*}
By the first rule,  $S$ is always propagated into the future from  $t$ to $t+2$. 
The second rule states that, if $t$ does not encode a cell and $\overline{N}$ holds therein,
 then $t+2$ does not encode a cell either.
By the third rule, if $N$ holds at $t$ and this $t$ is to the right of the time point encoding the initial state, then $t+2$ encodes an empty cell.
The last rule states that, if  $t$ encodes a cell with contents $s$ and the 
head is not on this cell, then $t+2$ also encodes a cell with contents $s$.

We next encode the left-moving transitions.
Proposition $L$ is used
to indicate a time point encoding a cell such that the 
head was on it in the previous configuration and then moved to the left.
Program $\Prog_\TM$  contains the following rules for every $s \in \Sigma$ and $q \in \states$ with transition $\delta(s,q)=(s',q',\Left)$, and every $s^* \in \Sigma$:
\begin{align*} 
\boxplus_2 L \land \boxplus_1 S_{q'} \land \boxplus_2 C_{s'} & \gets H \wedge C_{s} \wedge \diamondminus_{(0,2)} S_{q}, 
\\
 \bot &\gets L \wedge \boxminus_{(0,1)} \overline{H}, 
\\
\bot & \gets L \wedge \diamondminus_{(0,1)} (C_{s^*} \wedge \diamondminus_{(0,1)} H ).
\end{align*}
The first rule simulates the transition: $H$ holds as intended, the state is changed from $q$ to $q'$, and the contents of the cell under the head change from $s$ to $s'$ (the conjunction in the head is used for brevity and can be simulated by several rules).  
The last two rules encode the position of the head in the consequent configuration, 
by stating that $H$ holds at the first time point encoding a cell  to the left of the time point with $L$.

Similarly, for each transition ${\delta(s,q)=(s',q',\Right)}$ moving the head to the right and any $s^* \in \Sigma$, program $\Prog_\TM$ has the rules
 \begin{align*}
 \boxplus_1 S_{q'} \land \boxplus_2 C_{s'} &\gets H \wedge C_{s} \wedge \diamondminus_{(0,2)} S_{q},
 \\
\boxplus_2 H &\gets C_{s^*}\wedge \overline{C} \So_{(0,1)}  (H \wedge C_s) \wedge  \diamondminus_{(0,2)} S_{q},
\\
\boxplus_2 H &\gets N \wedge \overline{C} \So_{(0,1)}  (H \wedge C_s) \wedge  \diamondminus_{(0,2)} S_{q}.
\textbf{}
\end{align*}
Here, the first rule encodes the change of the state and the contents of the cell above which the head is.
The last two rules simulate the change of the position of the head.
Finally, $\Prog_\TM$ contains rule $\bot \gets S_{\qhalt}$, which yields
an inconsistency when the halting state is reached.
	
We next reduce an input word $\word=s_1 \dots s_{|\word|}$ to a dataset $\D_{\word}$. Assuming that $\word$ is non-empty, 
 we let ${t_k =  1 + \frac{k}{|\word|+1}}$ for each $k \in \{1,\dots, |\word|\}$ (it is only important here that ${1 < t_1< \dots < t_{|\word|} < 2}$); then,  $\D_{\word}$
contains the facts: 
\begin{align*}
&
S@[0,1],
\qquad
\overline{N}@[0,1],
\qquad
\overline{H}@[0,1],
\qquad
S_{\qinit} @ 1,
\qquad
H @ t_1,
\\
&
C_{s_1} @ t_1,
\; \ldots, \;
C_{s_{|w|}} @ t_{|w|},
\qquad
\overline{C}@(1,t_1),
\qquad
\overline{C}@(t_{1},t_{2}),
\;\; \ldots, \;\;
\overline{C}@(t_{|\word|},2).
\end{align*} 
Intuitively, $\D_{\word}$ describes the initial configuration
of $\TM$ on $\word$ within $[0,2)$; the
initial state is encoded in $1$ and  $t_1, \dots, t_{|\word|}$ encode the first $|\word|$ cells of $\TM$.
Moreover, $\overline{C}$ holds in all other time points in $(1,2)$, whereas $\overline{N}$ and $\overline{H}$ hold in $[0,1]$.

We next show that $\Prog_\TM$ and $\D_{\word}$ have a stable model if and only if $\TM$
does not halt on $\word$. Assume that $\Tmod$ is a stable model of $\Prog_{\TM}$ and $\D_{\word}$.
Then, using induction over $i \in \mathbb{N}$, we can prove 
that the $i$th configuration in the computation of $\TM$ on
$\word$ is encoded as discussed above.
In particular, if $q$ is the state of $\TM$ in the $i$th configuration,
then $\Tmod,t \models S_q$ for some $t \in [2i, 2i+1]$.
Since $\Tmod$ is a model of $\Prog_{\TM}$, however, it satisfies the rule $\bot \gets S_{\qhalt}$;
therefore, the state $\qhalt$ cannot occur in any configuration in the computation of $\TM$ on
$\word$, and so $\TM$ does not halt on $\word$.

For the opposite direction assume that 
$\TM$ does not halt on $\word$.
Then, we let $\Tmod$ be the minimal interpretation which satisfies 
$\D_{\word}$ 
and the following statements, for every positive integer $i$, where we let $t^{i}_j = 2i + 1+ \sum_{k=1}^j \frac{1}{2^k}$ for each $j \in \mathbb{N}$ (this definition ensures that there are infinitely many time points of the form $t_j^i$ in interval $(2i+1,2i+2)$):
\begin{itemize}[leftmargin=.3in]
\item[--] $\Tmod \models S @ [2i,2i+1]$,

\item[--] $\Tmod, t_j^i \models C_s$, whenever $s$ are the contents of the $j$th cell of $\TM$ in the  $i$th step of computation on $\word$, for $j \in \{1, \dots, |\word|+i\}$,

\item[--] $\Tmod,t \models \overline{C}$, for each $t \in (2i,2i+1) \setminus \{t^{i}_1, \dots, t^{i}_{|w|+i}\}$,

\item[--] $\Tmod, t_{|\word|+i+1}^{i} \models N$ and
$\Tmod, t' \models \overline{N}$
for each $t' \in [2i,2i+2) \setminus \{ t_{|\word|+i+1}^{i} \}$,

\item[--]  $\Tmod, t^i_j \models H$ if the head of $\TM$ is above the $j$th cell in the $i$th step of its computation on $\word$,

\item[--]  $\Tmod, t \models \overline{H}$ for each $t \in [2i,2i+2) \backslash \{t^i_j\}$, for $j$ such that the head of $\TM$ is above the $j$th cell in the $i$th step of its computation on $\word$,

\item[--]  $\Tmod, t^{i+1}_j \models L$ if the head of $\TM$ is above the $j$th cell in the $i$th step of its computation on $\word$,

\item[--] $\Tmod, t^i_j +1 \models S_q$ if $\Tmod, t^i_j \models H$ and $\TM$ is in the state $q$
in the $i+1$th step of computation on $\word$.

\end{itemize}
Clearly, $(\Tmod,\Tmod)$ is an HT-model of $\D_{\word}$, and
it can be also verified, by inspecting the rules in $\Prog_{\TM}$, that  $(\Tmod,\Tmod)$ is an HT-model of $\Prog_{\TM}$.
Finally, to show that $(\Tmod,\Tmod)$ is a stable model of $\Prog_{\TM}$ and $\D_{\word}$, we need to show that $\Tmod = \Hmod^{\Tmod}_{\Prog_{\TM},\D_{\word}}$. 
Towards this goal, we first construct
the sequence
$\Hmod_0$, $\Hmod_1$ of interpretations
as in the proof of Theorem \ref{thm::least},
for which it holds that ${\Hmod_{\omega_1} = \Hmod^{\Tmod}_{\Prog_{\TM},\D_{\word}}}$.
Then we can easily show, using transfinite induction, that for each
$i \in \mathbb{N}$ and each relational fact $M@t$ with $t \in [2i,2i+2)$,
$\Tmod \models M@t$ if and only if $\Hmod_i \models M@t$, and also that
${\Hmod_i  \models M@t}$ if and only if
$\Hmod^{\Tmod}_{\Prog_{\TM},\D_{\word}} \models M@t$,
which together imply $\Tmod = \Hmod^{\Tmod}_{\Prog_{\TM},\D_{\word}}$.
\end{proof}
 
Observe that the reduction above shows undecidability of reasoning in \MTLneg{} over the rational timeline, even if we restrict our attention to fixed propositional programs in the forward-propagating fragment, and we consider only bounded datasets. In the next section, we turn our attention to the integer timeline and show that reasoning becomes decidable.

\section{Decidability over the Integer Timeline}
\label{sec:integers}

In this section we consider the integer timeline and thus we
fix $\T = \Z$. 
We will show that, in this case, reasoning becomes decidable in
\EXPS{} with respect to data complexity;
furthermore, complexity drops to \PS{}
if we restrict our attention to
forward-propagating programs and datasets mentioning 
only bounded intervals---a setting 
well-suited for
stream reasoning \cite{WalegaAAAI,ronca2017stream}.
In this setting, the additional expressive power 
provided by stable models comes at no computational cost since
reasoning in the corresponding positive fragment
 is already \PS{}-complete~\cite{walega2020datalogmtl,DBLP:conf/ijcai/WalegaGKK19}.

In prior work on positive and stratifiable programs, 
upper bounds for reasoning
have been 
established by constructing generalised B\"{u}chi automata that
accept (words describing) models of a given program and dataset  \cite{walega2020datalogmtl,tena2021stratified}.
Checking existence of a stable model
is more demanding, as we additionally need to ensure
model minimality 
as in \Cref{def::stable}; this
 requirement is non-trivial, and we will handle it differently for the cases of arbitrary and forward-propagating programs.

In the general case (\Cref{general}), we 
construct two kinds of left- and right-moving automata:
the first kind allows us to check existence of an HT-model
of the form $(\Tmod,\Tmod)$, while the second kind allows us to 
check existence of an HT-model of the form $(\Hmod,\Tmod)$ with $\Hmod \neq \Tmod$.
Then, a pair of words $\wordL$ and $\wordR$ that are accepted, 
respectively, by a pair of left- and right-moving automata of the first kind, 
but not by any pair of left and right-moving automata of the second kind,
represents a stable model.
This construction is conceptually similar to
that of \citeN{cabalar2011automata} for  a logic with linear temporal operators and 
involves complementing nondeterministic automata, which leads to an
exponential blowup.
Consequently, we obtain an \EXPS{} upper bound and thus
an exponential gap in data complexity with respect to 
positive programs
\cite{walega2020datalogmtl}.
In the case of forward-propagating programs (\Cref{forward}) we propose a different construction exploiting
the fact that rules 
propagate information in a single temporal direction.
This allows us to build automata that guarantee model minimality without complementation.
As a result, we can establish a tight \PS{}  bound in data complexity.

\subsection{General Programs} \label{general}

It will be convenient for our presentation
to assume that programs are in a normal form analogous to that by  \citeN{tena2021stratified} for stratifiable programs:
in each normalised 
rule  the head is a relational atom or $\bot$,
there is neither nesting of metric operators nor occurrences of $\diamondminus$ or $\diamondplus$ in  rule bodies, and
the only unbounded interval allowed is
$[0, \infty)$.
\begin{proposition}\label{prop:normalisation}
Each program $\Prog$ can be normalised in polynomial time
 into a program $\Prog'$
such that, for each
dataset $\D$, program $\Prog$ and dataset $\D$ have a stable
 model if and only if so do $\Prog'$ and $\D$. 
\end{proposition}
\begin{proof}
To construct $\Prog'$ we first delete all (trivial) rules having $\top$ as the head.
Then,  we
eliminate from the remaining rule heads metric
operators (i.e., boxes).
To this end, we replace each rule of the form ${ \Box^1_{\varrho_1} \dots \Box^n_{\varrho_n} P(\sbf) \gets B}$, where $n\geq 0$, each $\Box^i$ is either $\boxminus$ or $\boxplus$, and $B$ is the body of the rule,
with  rules $P'(\sbf) \gets B$ and ${P(\sbf)  \gets \meddiamond^1_{\varrho_1} \dots \meddiamond^n_{\varrho_n} P'(\sbf)}$,
 where $\meddiamond^i = \diamondplus $ if $\Box^i = \boxminus$ and otherwise $\meddiamond^i = \diamondminus$, and  
 $P'$ is a fresh predicate of the same arity as $P$.

Second, we iteratively eliminate nested temporal operators from rule bodies.
To this end, consider a rule $r$ whose body has an  occurrence $M$ of a metric atom that mentions only one temporal operator and which is in the scope of some other temporal operator.
If $M$ is of one of the forms $\diamondminus_\varrho P(\sbf)$, $\diamondplus_\varrho P(\sbf)$, $\boxminus_\varrho P(\sbf)$, or $\boxplus_\varrho P(\sbf)$, then we replace it   with $P'(\sbf)$ and add a rule $P'(\sbf) \gets M$, where $P'$ is a fresh predicate of the same arity as $P$.
If $M$ mentions $\So$ or $\Uo$, we need to proceed in a special way to ensure safety of the new rules.
If $M$ is of the form $P_1(\sbf_1) \So_{\varrho} P_2(\sbf_2)$,  we remove $r$ and proceed as follows (note that the conditions below are not exclusive):
\begin{itemize}[leftmargin=.3in]
\item[--] if $0 \in \varrho$, we add the rule obtained by replacing $M$ in $r$ with $P_2(\sbf_2)$,
\item[--] if $1 \in \varrho$, we add the rule obtained by replacing $M$ in $r$ with $\boxminus_1 P_2(\sbf_2)$,
\item[--]
we add the rule  
${P'(\sbf_1,\sbf_2) \gets P_1(\sbf_1) \So_{\varrho} P_2(\sbf_2) \wedge \diamondminus_{[0, \infty)} P_1(\sbf_1)}$
and the rule obtained by replacing $M$ in $r$ with
$P'(\sbf_1, \sbf_2)$, 
where $P'$ is a fresh predicate whose arity
equals the sum of arities of $P_1$ and $P_2$.
\end{itemize}
If $M$ is of the form $P_1(\sbf_1) \Uo_{\varrho} P_2(\sbf_2)$, we proceed analogously to the case when $M$ is of the form $P_1(\sbf_1) \So_{\varrho} P_2(\sbf_2)$, but instead of $\boxminus$ and $\diamondminus$ we use, respectively, $\boxplus$ and $\diamondplus$ in the new rules.
Moreover, if  $M$ mentions $\top$ or $\bot$, we treat these symbols as nullary predicates
and proceed as before.

Next, we eliminate diamond operators by replacing $\diamondminus_\varrho M$ and $\diamondplus_\varrho M$ with, respectively, $\top \So_\varrho M$ and $\top \Uo_\varrho M$.
Finally, we eliminate unbounded intervals $\varrho$ different from $[0,\infty)$ as follows.
If a rule $r$ mentions
$\Box_{\varrho} P(\sbf)$, we replace this $\Box_{\varrho} P(\sbf)$
 with 
$\Box_{t} P(\sbf)$ and add a rule ${P'(\sbf) \gets \Box_{[0,\infty) }P(\sbf)}$,
where $\Box$ is either $\boxminus$ or $\boxplus$, $t$ is the least natural number in $\varrho$ (so $t \geq 1$), and $P'$ is a fresh predicate of the same arity as $P$.
In the case of operators $\So$ and $\Uo$ we need to pay special attention to ensure that the new
rules are safe. 
Assume that a rule $r$ mentions
${M = P_1(\sbf_1) \So_{\varrho} P_2(\sbf_2) }$ with an unbounded $\varrho \neq [0,\infty)$.
We remove  $r$ and proceed as follows, for $t$ the least natural number belonging to $\varrho$:
\begin{itemize}[leftmargin=.3in]
\item[--] if $1 \in \varrho$, we add the rule obtained by replacing $M$ in $r$ with $\boxminus_1 P_2(\sbf_2)$,
\item[--] if $1 \not\in \varrho$, we add the rule obtained by replacing $M$ in $r$ with $ \boxminus_{ t} P_2(\sbf_2) \land  \boxminus_{(0,t)} P_1(\sbf_1)$

\item[--] we add the rule  
$P'(\sbf_1, \sbf_2)  \gets P_1(\sbf_1) \So_{[0,\infty)} P_2(\sbf_2) \wedge 
P_1(\sbf_1)$,
and the rule obtained by replacing $M$ in $r$ with
$\boxminus_{ t  } P'(\sbf_1, \sbf_2) \land  \boxminus_{(0,t]} P_1(\sbf_1)$, 
where $P'$ is a fresh predicate whose arity
equals the sum of arities of $P_1$ and $P_2$.
\end{itemize}
In the case of atoms mentioning $\Uo$, as well as $\top$ or $\bot$, we proceed analogously.
\end{proof}

In the remainder of this section,
we fix a normalised program $\Pi$ and a dataset $\D$, and let 
 $\ground{\Prog}{\D}$ be the subset of $\groundp{\Prog}$ mentioning only
 constants from $\Prog$ and $\D$. 
Then, $\mat(\Prog, \D)$ is the set consisting of all relational atoms in $\D$, all metric atoms 
in $\ground{\Prog}{\D}$, and
all metric atoms of the forms $\boxminus_{ [ 0,\infty)} \matA$
and $\boxplus_{ [ 0,\infty)} \matA$, with $\matA$ a relational atom mentioned in 
 $\ground{\Prog}{\D}$. 

We next define the notion of a \emph{window}---a fragment of 
an  HT-interpretation over a particular interval; such
windows will serve as states of our automata.

\begin{definition}\label{def::window}
A \emph{window} is a tuple $(\varrho, H,T,b)$,
where $\varrho$ is a closed (and hence bounded) interval, $b \in \{ 0,1\}$, and
$H$ and $T$ are sets of metric
facts of the form $\matA @ t$ satisfying the following conditions: 
\begin{itemize}[leftmargin=.3in]
\item[--]  $\matA  \in\mat(\Prog,\D)$, $t \in
\varrho$, and $H \subseteq T$;
\item[--]  there exist interpretations $\Hmod$ and $\Tmod$ such that, for
each
${\matA \in \mat(\Prog,\D)}$ and  $t \in \varrho$,
\begin{itemize}
\item[--] ${\matA @ t \in H}$ if and only if $\Hmod \models \matA@t $, and 

\item[--] ${\matA @ t \in T}$ if and only if $\Tmod \models \matA@t $.
\end{itemize}
\end{itemize}
The window's \emph{length}  is  the length of~$\varrho$.
Finally, we say that a window is \emph{initial} if either  $H=T$ and $b=0$, or  $H \neq T$ and $b=1$.
\end{definition}

Intuitively, a window $(\varrho, H,T,b)$ is a fragment of an HT-interpretation $(\Hmod,\Tmod)$ restricted to  $\varrho$, 
where $H$ and $T$ describe facts  
holding within $\varrho$ in $\Hmod$ and  $\Tmod$, respectively.
Windows will serve as states of the automata recognising word representations of specific HT-interpretations, and 
in this process the flag $b$ is used to distinguish between stable and non-stable models; 
in particular,
our automata will ensure that flag $b$ is set to 
$1$ in each state $\window$ of a run such that  
$H \neq T$ in either $\window$ or in some previous state of this run.

By definition, a window can be extended to an HT-in\-ter\-pretation.
This HT-interpretation can be an HT-model of $\Prog$ only if the window locally satisfies $\Prog$, which we define next.

\begin{definition}
\label{def::locallySatisfies}
A window $(\varrho, H,T,b)$  \emph{locally satisfies}  $\Prog$
if, for each rule of Form~\eqref{eq:ruleneg}
in $\groundp{\Prog}$
and each $t \in \varrho$, both of the following hold:
\begin{enumerate}[leftmargin=.3in]
\item[--] $M@ t \in H$ if  $ M_i @ t \in H$ for each $i \in \{1, \ldots, k\}$ and $M_j @ t \notin T$ for each $j \in \{k+1, \ldots, m\}$, 
	
\item[--] $M@ t \in T$  if $ M_i @ t \in T$ for each $i \in \{1, \ldots, k\}$ and $M_j @ t \notin T$ for each $j \in \{k+1, \ldots, m\}$.
\end{enumerate}
\end{definition}

Next, given an initial window $\window_0$, we define automata  $\A_{\window_0}^\leftarrow$ and $\A_{\window_0}^\rightarrow$, which
will  allow us to recognise HT-models of $\Prog$ that extend $\window_0$.
In particular, 
if $\A_{\window_0}^\leftarrow$ and $\A_{\window_0}^\rightarrow$ 
accept words $\wordL$ and $\wordR$ respectively, then we will be able to construct an HT-model extending $\window_0$, for which the part
located to the left of $\window_0$ is described by  $\wordL$, and the part to the right
of $\window_0$ by~$\wordR$.

\begin{definition} \label{def::Buchi}
Let ${ \window_0=(\varrho_0,H_0,T_0,b_0 ) }$ be an initial window locally satisfying~$\Prog$.
Then, $\A_{\window_0}^\leftarrow = (\states, \Sigma, \delta, q_0,
\acc)$ is the generalised nondeterministic B\"{u}chi automaton with the following components:
\begin{enumerate}[leftmargin=.3in]
\item the states $\states$ consist of all windows of the
the same length as $\window_0$  locally satisfying $\Prog$;

\item the alphabet $\Sigma$ is the set of all $\sigma \subseteq \mat(\Prog, {\D})$;

\item the transition function $\delta : \states \times \Sigma \to 2^\states$ is such that
${(\varrho',H',T',b') \in \delta  \big( (\varrho,H,T,b), \sigma  \big)}$
if
\begin{itemize}
\item[--] ${\varrho'=[\varrho^- -1,\varrho^+ -1]}$,

\item[--] 
$\matA @ t \in H'$ if and only if $\matA @ t \in H$, for every ${\matA \in \mat(\Prog,\D)}$ and $t\in \varrho' \cap \varrho$,

\item[--] $T' =  \{ \matA @ t' \in T \mid  t'  \in  \varrho'\} \cup  \{ \matA  @  (\varrho^{-}{-}\,1) \mid \matA  \in  \sigma  \}$, and 

\item[--]  $b'=1$ whenever $b=1$ or $H' \neq T'$,  and $b'=0$ otherwise;
\end{itemize}

\item the initial state $q_0$  is $\window_0$;

\item the accepting condition $\acc$ is the family of sets of states which contains, for every atom 
${\boxminus_{[0,\infty)}  \matA \in \mat(\Prog,\D)}$, the  sets
\begin{align*}
\{  
(\varrho,H, T,b) \in \states
& \mid \text{there exists  $t \in \varrho$ such that } 
\boxminus_{[0,\infty)}  M @ t \in H  \text{ or }  M   @ t \notin H \},
\\
\{  
(\varrho,H, T,b) \in \states
& \mid \text{there exists  $t \in \varrho$ such that } 
\boxminus_{[0,\infty)}  M @ t \in T  \text{ or }  M   @ t \notin T \},
\end{align*}

and, for each $\matA_1 \So_{[0,\infty)}  \matA_2 \in \mat(\Prog,\D)$,
the sets
\begin{align*}
\{  (\varrho,H, T,b) \in \states 
& \mid  \text{there exists  $t \in \varrho$  such that } 
 \matA_1 \So_{[0,\infty)} \matA_2 @ t  \notin H \text{ or } \matA_2 @ t \in H \},
\\
\{  (\varrho,H, T,b) \in \states 
& \mid  \text{there exists  $t \in \varrho$  such that } 
 \matA_1 \So_{[0,\infty)} \matA_2 @ t  \notin T \text{ or } \matA_2 @ t \in T \}.
\end{align*}
\end{enumerate}
The   automaton $\A_{\window_0}^\to$  
is defined analogously, except that we let ${\varrho'= [\varrho^- +1, \varrho^+ +1] }$, in the definition of $T'$ we replace $\varrho^-{-}1$ with $\varrho^+ {+} 1$, 
and in the definition
of $\acc$ we use 
$\boxplus$ and $\Uo$
instead of 
$\boxminus$ and $\So$  , respectively.
\end{definition}

Accepting runs of these automata will correspond to HT-interpretations.
Indeed, as we will show next,  each HT-interpretation can be decomposed into an infinite sequence of windows
$ \dots , \window_{-1}, \window_0, \window_1,\dots$
such that 
$\window_{0},\window_{-1}, \dots$ and
 $\window_0, \window_1,\dots$   are accepting runs of
 $\A_{\window_0}^\gets$ and
$\A_{\window_0}^\to$, respectively.
We define the decomposition of an HT-interpretation as follows.

\begin{definition}\label{def:decomposition}
We define the $\varrho$-\emph{decomposition} of an HT-interpretation $(\Hmod, \Tmod)$, for a bounded interval $\varrho$, as the sequence of tuples 
${\window_i= (\varrho_i, H_i, T_i, b_i)}$, for  $i\in\Z$, such that the following hold:
\begin{itemize}[leftmargin=.3in]
\item[--] ${\varrho_i = [\varrho^- +i, \varrho^+ +i]}$,
\item[--] $H_i$ is the set of all facts $\matA @ t$ such that $\matA \in \mat\PD$, $t \in  \varrho_i$, and
 ${\Hmod  \models  \matA @ t}$, 
\item[--] $T_i$ is the set of all facts $\matA @ t$ such that $\matA \in \mat\PD$, $t \in  \varrho_i$, and
 ${\Tmod  \models  \matA @ t}$, 
\item[--] $b_i=1$ if 
there exists $j \in \{0, \dots, i \} $ such that $H_j \neq T_j$; otherwise, $b_i=0$.
\end{itemize}
\end{definition}

\begin{lemma}\label{lem:decomp}
Let $\varrho$ be a bounded interval, let a sequence of tuples ${\window_i= (\varrho_i, H_i, T_i, b_i)}$, with $i \in \Z$, be a   $\varrho$-\emph{decomposition} of an HT-model $(\Hmod, \Tmod)$ of $\Prog$,  and
let  ${\wordL= \sigma_{-1} \sigma_{-2} \cdots }$ and  ${\wordR= \sigma_{1} \sigma_{2} \cdots}$  be the words such that 
$\sigma_k = T_k \setminus T_{k+1}$ for $k<0$
and
$\sigma_k=T_k \setminus T_{k-1} $ for ${k>0}$.
Then the following hold:
\begin{enumerate}[leftmargin=.3in]
\item each $\window_i$ is a window locally satisfying $\Prog$, 
\item $\window_0, \window_{-1}, \dots$ is an accepting run of $\A_{\window_0}^\gets$ on $\wordL$, and
\item $\window_0, \window_{1}, \dots$ is an accepting run of $\A_{\window_0}^\to$ on $\wordR$.
\end{enumerate}
\end{lemma}
\begin{proof}
To show that Statement 1 holds, we start by
observing that each $\window_i$ satisfies all the conditions from \Cref{def::window} of a window.
Indeed, by \Cref{def:decomposition}, $H_i$ and $T_i$ are sets of metric
facts $\matA @ t$ with $\matA  \in\mat(\Prog,\D)$ and $t \in
\varrho_i$ such that
$H_i \subseteq T_i$, whereas  $b_i \in \{0,1 \}$;
moreover, $\Hmod$ and $\Tmod$ witness existence of the interpretations required in \Cref{def:decomposition}. 
Furthermore, since $(\Hmod, \Tmod)$ is an HT-model of $\Prog$,  each $\window_i$ locally satisfies $\Prog$.

To prove that Statement 2, we observe that $\window_0$ is an initial window by construction, so the automaton $\A_{\window_0}^\leftarrow = (\states, \Sigma, \delta, q_0,
\acc)$ is well-defined.
By the definition of the transition function  $\delta$ of $\A_{\window_0}^\gets$ (see \Cref{def::Buchi}) as well as by the construction of $\window_0, \window_{-1}, \dots$ and $\wordL = \sigma_{-1} \sigma_{-2} \dots$,
we get that  ${\window_{i-1} \in \delta(\window_i, \sigma_{i-1})}$, for each  integer $i \leq 0$.
Thus, $\window_0, \window_{-1}, \dots$ is a run of $\A_{\window_0}^\gets$ on $\wordL$.
It remains to show that this run is accepting, that is, 
for every set  $S$ in the accepting condition $\acc$, there are infinitely many  integers $i < 0$ such that $\window_i \in S$.
Towards a contradiction suppose that there exists $S \in \acc$ which does not satisfy this property. Thus, there exists  $i \leq 0$ such that $\window_j \notin S$, for all $j \leq i$.
Assume first that $S$ is of one of the following forms
\begin{align*}
\{  
(\varrho',H, T,b) \in \states
& \mid \text{there exists  $t \in \varrho'$ such that } 
\boxminus_{[0,\infty)}  M @ t \in H  \text{ or }  M   @ t \notin H \}, \text{ or}
\\
\{  
(\varrho',H, T,b) \in \states
& \mid \text{there exists  $t \in \varrho'$ such that } 
\boxminus_{[0,\infty)}  M @ t \in T  \text{ or }  M   @ t \notin T \},
\end{align*}
for some ${\boxminus_{[0,\infty)}  \matA \in \mat(\Prog,\D)}$.
As a result,  $\I \not\models \boxminus_{[0,\infty)}  M @ t $ and  $\I \models M @ t $,
for each $t \leq \varrho_i^-$, where  the interpretation $\I$ is either $\Hmod$ or $\Tmod$, depending whether $S$ is of the first or the second form presented above.
This, however, contradicts the semantics of $\boxminus_{[0,\infty)}$.
It remains to consider the case where $S$ is of one of the following forms 
\begin{align*}
\{  (\varrho',H, T,b) \in \states 
& \mid  \text{there exists  $t \in \varrho'$  such that } 
 \matA_1 \So_{[0,\infty)} \matA_2 @ t  \notin H \text{ or } \matA_2 @ t \in H \}, \text{ or}
\\
\{  (\varrho',H, T,b) \in \states 
& \mid  \text{there exists  $t \in \varrho'$  such that } 
 \matA_1 \So_{[0,\infty)} \matA_2 @ t  \notin T \text{ or } \matA_2 @ t \in T \},
\end{align*}
for some $\matA_1 \So_{[0,\infty)}  \matA_2 \in \mat(\Prog,\D)$.
Then, 
$\I \models \matA_1 \So_{[0,\infty)} \matA_2 @ t $ and  ${\I \not\models \matA_2 @ t  }$,
for each $t \leq \varrho_i^-$, where
$\I$ is either $\Hmod$ or $\Tmod$, depending whether $S$ is of the first or the second form presented above.
This directly contradicts the semantics of $\So_{[0,\infty)}$.
Consequently, $\window_0, \window_{-1}, \dots$ needs to be an accepting run of $\A_{\window_0}^\gets$ on $\wordL$.

The proof of Statement 3 is analogous to the proof of Statement 2, due to the symmetry between the automata $\A_{\window_0}^\gets$ and $\A_{\window_0}^\to$.
\end{proof}

To check existence of a stable model, however, we require automata 
that recognise HT-models $(\Hmod,\Tmod)$ with $\Hmod=\Tmod$, and automata 
that  recognise HT-models $(\Hmod,\Tmod)$ with $\Hmod \subsetneq \Tmod$.
The intersection of the former with the complement of the latter allows us to recognise stable models---that is, essentially, 
HT-models $(\Tmod,\Tmod)$ for  
which there are no models $(\Hmod,\Tmod)$ with $\Hmod \subsetneq \Tmod$.
To this end, we define two more types of automata as follows:

\begin{definition}\label{defBC}
Let $\window_0=(\varrho_0,H_0,T_0,b_0)$ be an initial window locally satisfying~$\Prog$.
We define non-deterministic generalised B\"{u}chi automata $\B_{\window_0}^\leftarrow$, $\B_{\window_0}^\rightarrow$ and
$\C_{\window_0}^\leftarrow$, $\C_{\window_0}^\rightarrow$ as follows:
\begin{enumerate}[leftmargin=.3in]
\item[--] if $H_0 = T_0$ and  $b_0 = 0$, the automata $\B_{\window_0}^\leftarrow$ and $\B_{\window_0}^\rightarrow$
are defined as $\A_{\window_0}^\leftarrow$ and $\A_{\window_0}^\rightarrow$, respectively,
except that for a window $(\varrho,H,T,b)$ to be a state we additionally require that $H=T$,

\item[--] the automata $\C_{\window_0}^\leftarrow$ and $\C_{\window_0}^\rightarrow$
are defined as $\A_{\window_0}^\leftarrow$ and $\A_{\window_0}^\rightarrow$, respectively,
except that we add to the accepting condition the set 
$\{ (\varrho,H, T,b) \in \states  \mid b=1\}$.
\end{enumerate}
\end{definition}

Intuitively, if $\window_0$ satisfies $\D$, then the automata $\B_{\window_0}^\gets$ and $\B_{\window_0}^\to$ 
recognise interpretations $\Tmod$ such that
 $(\Tmod,\Tmod)$ is an HT-model of   $\Prog$ and $\D$.
Furthermore, interpretations $\Tmod$ accepted by
 $\A_{\window_0}^\gets$ and $\C_{\window_0}^\to$, 
or by  $\C_{\window_0}^\gets$ and $\A_{\window_0}^\to$
are such that $(\Hmod,\Tmod)$ is an HT-model of $\Prog$ and $\D$, for some  $\Hmod \subsetneq \Tmod$.
Hence, as we show  next, we can use these automata to recognise stable models.
This, however, requires the additional assumption
that the  
windows in the automata are long enough to allow for capturing the semantics of metric operators occurring in a program.
To this end, we will use windows of the same length as the interval 
${\varrho_{\PD} = [t^{\min}_{\D},t^{\max}_{\D}+t_{\Prog}]}$.
If the length of an initial window $\window_0$ is as required, then we can reduce checking existence of a stable model to checking particular properties of our automata, as stated next.

\begin{theorem} \label{stable_reduction}
Program $\Prog$ and dataset $\D$ have a stable model 
if and only if 
there exists an initial window 
${\window_0=(\varrho_0,T_0, T_0, 0)}$ locally satisfying $\Prog$ with ${\varrho_0 = \varrho_{\PD}}$ and
$T_0 \models \D$, and there exist words $\wordL$ and $\wordR$ over $2^{\mat(\Prog, {\D})}$ such that
both of the following hold:
\begin{enumerate}[leftmargin=.3in]
\item $\wordL$ and $\wordR$ are accepted by   $\B_{\window_0}^\gets$ and $\B_{\window_0}^\to$, respectively, 

\item there is no initial window 
${\window'_0=(\varrho_0,H_0, T_0, b_0)}$ locally satisfying $\Prog$ such that $H_0 \models \D$,  
 and 
$\wordL$ and $\wordR$ are accepted either
by $\C_{\window'_0}^\gets$ and $\A_{\window'_0}^\to$, respectively, or by 
$\A_{\window'_0}^\gets$ and $\C_{\window'_0}^\to$, respectively.
\end{enumerate}
\end{theorem}

\begin{proof}
Assume that $\Tmod$ is a stable model of $\Prog$ and $\D$.
We will show how to construct the required $\window_0$,  $\wordL$, and $\wordR$.
To this end, let 
$\dots, \window_{-1}, \window_0, \window_1, \dots$
be the $\varrho_{\PD}$-decomposition of $(\Tmod, \Tmod)$ with ${\window_i= (\varrho_i, H_i, T_i, b_i)}$, for each $i\in\Z$.
By Definition \ref{def:decomposition}, we obtain that $H_i=T_i$ and $b_i=0$, for all $i \in \Z$; furthermore, $\window_0$ locally satisfies $\Prog$ by Lemma \ref{lem:decomp}.
Finally, we have ${\varrho_0 = \varrho_{\PD}}$, which ensures
$T_0 \models \D$, so $\window_0$ satisfies the initial requirements from the theorem.
Next, let ${\wordL= \sigma_{-1} \sigma_{-2} \cdots}$ and $\wordR= \sigma_{1} \sigma_{2} \cdots$ be the words such that 
$\sigma_k=T_k \setminus T_{k+1}$ if $k<0$,
and $\sigma_k=T_k \setminus T_{k-1} $ if ${k>0}$.
It remains to show that $\window_0$, $\wordL$, and $\wordR$ satisfy 
 Conditions 1 and 2 from the theorem.
 
To show that Condition 1 holds,
we observe that, by \Cref{lem:decomp}, 
$\window_0, \window_{-1}, \dots$ is an accepting run of
$\A_{\window_0}^\gets$
on $\wordL$,
and
$\window_0, \window_{1},  \dots$ is an accepting run of
$\A_{\window_0}^\to$ on $\wordR$.
Moreover, since $H_i=T_i$ and $b_i = 0$ for all $i \in \Z$,
we obtain that $\window_0, \window_{-1}, \dots$ is an accepting run of
$\B_{\window_0}^\gets$
on $\wordL$,
and
$\window_0, \window_{1},  \dots$ is an accepting run of
$\B_{\window_0}^\to$ on $\wordR$.

Next, let us  suppose towards a contradiction that there exists ${\window_0' =(\varrho_0',H_0', T_0', b_0') }$
witnessing a violation of Condition 2.
Hence, $\window_0'$ is an initial window locally satisfying $\Prog$ such that $\varrho_0'=\varrho_0$, $T_0' = T_0$, and $H_0' \models \D$.
Moreover, we assume  that $\wordL$ and $\wordR$ are accepted 
by $\C_{\window'_0}^\gets$ and $\A_{\window'_0}^\to$, respectively.
Hence,  $\C_{\window'_0}^\gets$ has an accepting run $\window'_0, \window'_{-1}, \dots$ on $\wordL$
and
$\A_{\window'_0}^\to$ has an accepting run $\window'_0, \window'_{1}, \dots$ on $\wordR$, 
where we let ${\window'_i= (\varrho_i', H_i', T_i', b_i')}$, for all $i \in \Z$.
Clearly, $\varrho_i' = \varrho_i$, for all $i \in \Z$.
Moreover, by the definition of the transition functions of the automata and the construction of $\wordL$ and $\wordR$, we obtain that $T_i' = T_i$, for all $i \in \Z$.
Therefore,   $\Tmod$ is the least model of all relational facts in   $\bigcup_{i\in \Z}T'_i$.
We let $\Hmod$ be the least model of all relational facts in   $\bigcup_{i\in \Z}H_i$; we will show that 
$(\Hmod, \Tmod)$ is an HT-model of $\Prog$ and $\D$.

Since $H_i' \subseteq T_i'$, for all $i \in \Z$, we obtain that $\Hmod \subseteq \Tmod$, and so $(\Hmod, \Tmod)$  is an HT-interpretation.
Moreover, $(\Hmod, \Tmod)$ is an HT-model of $\D$, as $H_0' \models \D$.
Next we will show that $(\Hmod, \Tmod)$ is an HT-model of $\Prog$.
Since each $\window_i'$ is a window of length  $[t^{\min}_{\D},t^{\max}_{\D}+t_{\Prog}]$ (which is the length of $\varrho_{\PD}$), it holds that 
$\Hmod \models M@t $ if and only if $M@t \in H_i'$;
as well as
$\Tmod \models M@t $ if and only if $M@t \in T_i'$,
for any $M \in \mat(\Prog,\D)$ and $t \in \varrho_i'$.
Indeed, we have shown an analogous statement  for positive~\cite[Lemma 9]{DBLP:conf/ijcai/WalegaGKK19} and stratifiable programs~\cite{tena2021stratified}, and the same argument applies here.
Now, to show that $(\Hmod,\Tmod)$ is a model of $\Prog$, we fix a ground rule from $\mathsf{ground}(\Prog,\D)$  of Form~\eqref{eq:ruleneg} and a time point
$t$.
If 
$\Hmod ,t \models M_i$ for all ${i \in \{1, \ldots, k\} }$ and ${\Tmod ,t \not\models  M_j}$ for all ${j \in \{k+1, \ldots, m\} }$, 
then $M_i@t \in H_n'$ for all $i \in \{1, \ldots, k\}$ and $M_j@t \notin T_n$ for all $j \in \{k+1, \ldots, m\}$, for each $n$ such that $t \in \varrho_n'$.
Since $\window_n'$ is a window locally satisfying $\Prog$,
we obtain  that $M@t \in H_n$ (where $M$ is the head of $r$),
which implies $\Hmod ,t \models M$ by definition of $\Hmod$ and the fact
that $\Prog$ is in normal form so $M$ is a relational atom.
Similarly, if $\Tmod ,t \models M_i$ for all $i \in \{1, \ldots, k\}$ and ${\Tmod ,t \not\models  M_j}$ for all $j \in \{k+1, \ldots, m\}$, then we obtain that $\Tmod ,t \models M$.
Hence, $(\Hmod,\Tmod)$ is indeed an HT-model of $\Prog$.

However,  the  accepting condition of $\C_{\window'_0}^\gets$ guarantees that $b'_i=1$, for some $i \leq 0$, and so ${H_i' \subsetneq T_i' }$.
Therefore, $\Hmod \subsetneq \Tmod$, and so $\Tmod$ is not a stable model, which rises a contradiction.
If $\wordL$ and $\wordR$ are accepted 
by 
$\A_{\window'_0}^\gets$ and $\C_{\window'_0}^\to$, respectively, then we  construct in an analogous way a run $\window'_0, \window'_{-1}, \dots$ of $\A_{\window'_0}^\gets$ and a run 
$\window'_0, \window'_{1}, \dots$ of 
$\C_{\window'_0}^\to$.
Repeating the argumentation above, we construct an HT-model $(\Hmod,\Tmod)$ of $\Prog$ and $\D$, and then, we show that there exists $i\geq 0$ such that $b'_i=1$.
Thus, $H_i' \subsetneq T_i'$, so
$\Hmod \subsetneq \Tmod$, and consequently, $\Tmod$ is not a stable model, raising again a contradiction.
Thus, Condition 2 holds.

For the converse implication, assume that there exist $\window_0$, $\wordL$,
and $\wordR$ as described in the statement of the theorem. By Condition 1,
there is an accepting run $\window_0, \window_{-1}, \dots$ 
of $\B_{\window_0}^\gets$ on $\wordL$, and an accepting run $\window_0, \window_1, \dots$ of $\B_{\window_0}^\to$ on $\wordR$,  where ${\window_i= (\varrho_i, T_i, T_i, 0)}$.
We argue that the least model $\Tmod$ of relational facts in $\bigcup_{i\in\Z} T_i$ is a stable model of $\Prog$ and $\D$.
By an argument analogous to the second to last paragraph above,
$(\Tmod,\Tmod)$ is an HT-model of $\Prog$ and $\D$.
Suppose for contradiction that $\Tmod$ is not stable, so there exists
$\Hmod \subsetneq \Tmod$ such that
 $(\Hmod, \Tmod)$ is an HT-model of $\Prog$ and $\D$.
Let $\dots, \window'_{-1}, \window'_0, \window'_1, \dots$ be the 
$\varrho_{\PD}$-decomposition of $(\Hmod,\Tmod)$, with
${\window_i'= (\varrho_i', H_i', T_i', b_i')}$, and observe that
$\window'_0$ is an initial window.
By Lemma \ref{lem:decomp},
$\window'_0, \window'_{-1}, \dots$ is an accepting run of $\A_{\window'_0}^\gets$ on $\wordL$, and  
$\window'_0, \window'_{1}, \dots$ is an accepting run of $\A_{\window'_0}^\to$ on $\wordR$.
Moreover, since $\Hmod \subsetneq \Tmod$,
 there is $i \in \Z$ such that $H_i' \neq T_i'$, and so $b_i'=1$.
If $i \leq 0$, then $b_j' =1$ for all $j \leq i$, and so $\C_{\window'_0}^\gets$ accepts $\wordL$; analogously, if $i \geq 0$, then $\C_{\window'_0}^\to$ accepts $\wordR$.
Thus, Condition~2 does not hold, leading to a contradiction.
\end{proof}

\Cref{stable_reduction}  reduces checking
existence of a stable model  
to checking specific properties of our automata.
We aim at showing that the latter is feasible in \EXPS{}.
The main obstacle, however, is the size of states in
the automata: $\window_0$ is exponential in size with 
respect to $\D$,
and states of the automata from \Cref{stable_reduction} can be arbitrarily large 
since time points in windows can be arbitrary integers.
To remedy the first issue, we 
let $t_\Prog$ be the largest positive number mentioned in $\Prog$,
and we let $t_{\Prog}=1$ if $\Prog$ mentions no positive numbers---this choice of value is arbitrary since 
in this case we only need $t_\Prog$ to be a positive number in the timeline. Then,
we show in \Cref{short} that it suffices to consider automata with
states (i.e., windows) of length $t_\Prog$, which does not depend on $\D$.
The second issue is addressed by \Cref{lem:equiv_auto} which tells us that, 
rather than
considering automata with states of unbounded size (each of length ${t_\Prog}$), we can construct equivalent automata 
with polynomial-size states instead.

To state \Cref{short}, we define 
the left-most and right-most fragments of length $t_\Prog$ of a window as follows.

\begin{definition}\label{def:LR}
For a window ${\window=(\varrho,H,T,b)}$ of length at least $t_{\Prog}$,
we define $\window^{L}=( \varrho^L, H^L,T^L, b^L )$, where
\begin{itemize}[leftmargin=.3in]
\item[--]
  ${\varrho^L=[\varrho^-,\varrho^- + t_\Prog]}$,

\item[--] ${H^L = \{M@t \in  H \mid t \in \varrho^L \} }$,
\item[--]
${T^L = \{M@t \in  T \mid t \in \varrho^L \} }$,
\item[--] $b^L = 1$ if ${H^L \neq T^L}$; otherwise, $b^L = 0$.
\end{itemize}
Analogously, we let 
$\window^{R}=( \varrho^R, H^R,T^R, b^R )$, where 
\begin{itemize}[leftmargin=.3in]
\item[--]  ${\varrho^R=[\varrho^+ - t_\Prog, \varrho^+] }$,
\item[--] ${H^R = \{M@t \in  H \mid t \in \varrho^R \} }$,
\item[--] ${T^R = \{M@t \in  T \mid t \in \varrho^R \} }$,
\item[--] $b^R = 1$ if  ${H^R \neq T^R}$; otherwise, $b^R = 0$.
\end{itemize}
\end{definition}

We observe that if $\window$ is a window locally satisfying $\Prog$, then
$\window^L$ and $\window^R$ (which are ``fragments'' of
$\window$)
also are windows locally satisfying $\Prog$.
Furthermore, both $\window^L$ and $\window^R$  are initial windows by definition.
Thus, if $\window_0$ is such that
$X_{\window_0}^\gets$ and $X_{\window_0}^\to$ are well-defined automata, for some
$X \in \{ \A,\B,\C \}$,
then
$X_{\window_0^{L}}^\gets$  and $X_{\window_0^{R}}^\to$ are also well-defined automata.
Moreover, we can show several equivalences between these automata.

\begin{lemma}\label{short}
Let $X \in \{ \A,\B,\C \}$
and let $\window_0 = (\varrho_0,H_0,T_0,b_0)$
be a window  such that $X_{\window_0}^\gets$ and $X_{\window_0}^\to$ are well-defined. Then, the following statements hold:
\begin{enumerate}[leftmargin=.3in]
\item 
If $X  \in  \{\A, \B\}$, then 
$X_{\window_0}^\gets$ and $X_{\window_0}^\to$ are equivalent to
$X_{\window_0^{L}}^\gets$  and $X_{\window_0^{R}}^\to$, respectively.

\item 
If $X   = \C$ and $H_0 = T_0$, then $X_{\window_0}^\gets$ and $X_{\window_0}^\to$ are equivalent to
$X_{\window_0^{L}}^\gets$  and $X_{\window_0^{R}}^\to$, respectively.

\item  If $X   = \C$ and $H_0 \neq T_0$, then 
$X_{\window_0}^\gets$ and 
$X_{\window_0}^\to$
are equivalent to
$\A_{\window_0^{L}}^\gets$ and $\A_{\window_0^{R}}^\to$, respectively. 
\end{enumerate}
\end{lemma}
\begin{proof}
First, we show Statement 1 for $X = \A$.
Assume that $\window_0, \window_{-1}, \dots$ is an accepting run of $\A_{\window_0}^\gets$
on a word $\word =\sigma_{-1},\sigma_{-2},\dots$. We will show that there
is an accepting run $\window'_0, \window'_{-1}, \dots$ of $\A_{\window^L_0}^\gets$ on the same word.
To define this run, for each integer $i \leq 0$, we let 
${\window^L_i= (\varrho_i,H_i,T_i,b_i)}$, and 
we define $\window'_i = (\varrho_i, H_i,T_i,b'_i)$, where $b'_i = 1$ if there exists $j \in \{0, \dots, i \}$ such that $H_i \neq T_i$, and otherwise $b'_i=0$. 
In particular, $\window_0' = \window_0^L$.
Now, we will show that $\window'_0, \window'_{-1}, \dots$ is an accepting run of
$\A_{\window_0^L}^\gets$ on  $\word$.
We start by observing that since each $\window_i$ is a window locally satisfying $\Prog$ and since
$\window_i^L$ is a ``fragment'' of $\window_i$, then $\window'_i$ is also a window locally satisfying $\Prog$.
Thus, $\window'_0, \window'_{-1}, \dots$ are states of $\A_{\window_0^L}^\gets$.
To show that they constitute a  run of $\A_{\window_0^L}^\gets$ on  $\word$,
we observe that the transition functions $\delta$ of $\A_{\window_0}^\gets$ and $\delta^L$ of $\A_{\window_0^L}^\gets$ are the same modulo definition of states.
This observation and our definition of $b'_i$,
 ensure that $\window_{i-1} \in \delta(\window_i,\sigma_{i-1})$ implies
$\window'_{i-1} \in \delta^L(\window'_i,\sigma_{i-1})$, for any integer $i \leq 0$.
It remains to show that $\A_{\window_0^L}^\gets$ accepts the run $\window'_0, \window'_{-1}, \dots$.
For this, we observe that $\A_{\window_0}^\gets$ and $\A_{\window_0^L}^\gets$ have the same accepting conditions modulo the definition of states.
Next, let $k$ be the difference between lengths of windows $\window_0$ and $\window^L_0$.
Observe that for any integer $i \leq -k$, 
if $\window_i$ belongs to some set $S$ from the accepting condition
of $\A_{\window_0}^\gets$, then there exists
$j \in [i,i+k]$ such that $\window'_j$ belongs to a set from the accepting condition
of $\A_{\window_0^L}^\gets$ that corresponds to $S$.
Since each set in the accepting condition of $\A_{\window_0}^\gets$ is visited
infinitely often by the states in the run $\window_0, \window_{-1}, \dots$,
 each set in the accepting condition of $\A_{\window_0^L}^\gets$ 
 is also visited infinitely often by the states in the run $\window'_0, \window'_{-1}, \dots$.
Thus, the latter is an accepting run of $\A_{\window_0^L}^\gets$.
The case of $\A_{\window_0^R}^\to$ is symmetric  so, by an analogous
argumentation, we obtain that if $\A_{\window_0}^\to$ accepts a word $\word$,
then so does $\A_{\window_0^R}^\to$.

For the opposite direction of Statement 1,
let us assume that $\A_{\window_{0}^L}^\gets$ has an
accepting run $\window_{0}^L, \window_{-1}', \window_{-2}',\dots$ 
on a word $\word = \sigma_{-1},\sigma_{-2},\dots$; we will show that 
 $\A_{\window_0}^\gets$ has an accepting run on 
 $\word$. To facilitate the definition of this run,
we divide $\window_0 = (\varrho_0,H_0,T_0,b_0)$ into tuples
$\window_{0}', \dots, \window_{k}'$,
where $k$ is the difference between the lengths of $\window_0$ and $\window_0^L$.
In particular, for each $i \in \{0, \dots, k \}$, we define 
${\window_{i}' = (\varrho_i',H_i',T_i',b_i') }$, such that
${\varrho_i'=[\varrho_0^- + k , \varrho_0^- +t_{\Prog} + k ]}$,
${H_i' = \{M@t \in  H_0 \mid t \in \varrho_i' \} }$,
${T_i' = \{M@t \in  T_0 \mid t \in \varrho_i' \} }$,
whereas $b_i' = 1$ if $H'_i \neq T'_i$ and $b'_i=0$ if $H'_i = T'_i$.
Recall that since $\A_{\window_0}^\gets$ is well-defined,
 $\window_0$ is an initial window locally satisfying $\Prog$, and
as we already mentioned, by fragmenting windows locally satisfying $\Prog$
we obtain windows which also locally satisfy $\Prog$.
Hence all of
$\window_{0}', \dots, \window_k'$ locally satisfy $\Prog$.
Furthermore, observe that $\window_{0}' = \window_{0}^L$.
Next, we show how to merge the windows
$\window_k', \dots, \window_0', \window_{-1}',\dots$ to obtain  an accepting run of $\A_{\window_{0}}^\gets$  on $\word$.
To this end, for arbitrary windows $\window = ( \varrho, H,T, b )$ and ${\window' =( \varrho', H',T', b' )}$ such that $\varrho \cap \varrho' \neq \emptyset$, we define their \emph{union}, written as $\window \cup \window'$, as the tuple ${(\varrho \cup \varrho', H \cup H', T \cup T', \max(b,b'))}$.
We can then establish the following claim:
\begin{claim}\label{claim:union}
Let $\window = (\varrho, H, T, b)$ and $\window' = (\varrho',H', T', b')$ be windows such that the left and right endpoints of $\varrho'$ succeed those 
of $\varrho$, and for each $t \in \varrho \cap \varrho'$ and each $X \in \{H,T\}$,
$M@t \in X$ if and only if $M@t \in X'$.
If both $\window$ and $\window'$ are of length at least $t_{\Prog}$ and locally satisfy $\Prog$, then
$\window \cup \window'$ is a window locally satisfying $\Prog$.
\end{claim}
\begin{proof}[Proof of \Cref{claim:union}]
Let $\window'' = \window \cup \window'  = (\varrho'', H'', T'', b'')$.
If $\window$ and $\window'$ locally satisfy $\Prog$, then clearly $\window''$ also locally 
satisfies $\Prog$. It remains to show that $\window''$ is a window.
To this end, for $\I$ an interpretation, $\varrho_b$ a closed interval, and $J$ 
a set of metric facts $M@t$ with $M \in \mat(\Prog,\D)$ and $t \in \varrho_b$, 
we say that $\I$ \emph{corresponds} to $J$ over $\varrho_b$ if, for every $t \in \varrho_b$ and
$M \in \mat(\Prog,\D)$, it holds that $\I \models M@t$ if and only if $M@t \in J$.  
To show that $\window''$ is a window, the only non-trivial
condition that needs to be verified is that there exist two interpretations that correspond, 
respectively, to $H''$ and $T''$ over $\varrho''$. We show this for $H''$ only, as the argument for $T''$ is analogous.
 Let $\Hmod$ (resp.\ $\Hmod'$) be an interpretation corresponding to $H$ over $\varrho$ 
 (resp.\ $H'$ over $\varrho'$); such an interpretation exists, since $\window$ (resp.\ $\window'$) is  
 a window. Let $\Hmod''$ be an arbitrary interpretation that coincides with $\Hmod$ 
 over $\varrho$  and all time points to its left, and with $\Hmod'$
 over $\varrho'$ and all time points to its right. 
Note that $\Hmod$ and $\Hmod'$ agree
  over $\varrho \cap \varrho'$ since, by assumption, $M@t \in H$ if and only if $M@t \in H'$
  for all $t \in \varrho \cap \varrho'$; thus, $\Hmod''$ is well-defined.
 It remains to show that $\Hmod''$ corresponds to $H''$ over $\varrho''$---that is,
 for each $M@t$ with $M \in \mat(\Prog,\D)$ and $t \in \varrho''$,
 we must show that $\Hmod'' \models M @ t$ if and only if $M @ t \in H''$.
 We show this explicitly for the case where $M$ is of the form $\boxminus_{\varrho_b} P$, 
 for $\varrho_b$ a bounded interval and $P$ a relational atom; the remaining cases can be proved in a similar
 way; see \cite{DBLP:conf/ijcai/WalegaGKK19}, Lemma 7. Suppose that $t \in \varrho$.
Since $\Hmod''$  and $\Hmod$ coincide over $\varrho$ and all time points to its left, we have that 
  $\Hmod'' \models \boxminus_{\varrho_b} P @t$ if and only if $\Hmod \models \boxminus_{\varrho_b} P@t$.
 Then, since $\Hmod$ corresponds to $H$ over $\varrho$, it holds that $\Hmod \models \boxminus_{\varrho_b} P @ t$ if and only if $\boxminus_{\varrho_b} P @t \in H$. Finally, since
  $H$ and $H''$ coincide over $\varrho$, $\boxminus_{\varrho_b} P @t\in H$ if and only if $\boxminus_{\varrho_b} P @ t \in H''$. Thus, by chaining together these double implications, we have that
   $\Hmod'' \models \boxminus_{\varrho_b} P @t$ if and only if
    $\boxminus_{\varrho_b} P @ t \in H''$, which is the target equivalence.
 Suppose now that $t \notin \varrho$, so $t$ is the right endpoint of $\varrho'$.
 By the semantics of metric atoms, $\Hmod'' \models \boxminus_{\varrho_b} P @ t$ if and only if $\Hmod'' \models P @ t-t'$ for each $t' \in \varrho_b$. Since the length of $\varrho'$ is at least $t_{\Prog}$, we have
 $t - t' \in \varrho'$ for each $t' \in \varrho_b$. Hence, since $\Hmod''$ and $\Hmod'$ coincide over $\varrho'$, $\Hmod'' \models\boxminus_{\varrho_b}  P$ if and only if $\Hmod' \models\boxminus_{\varrho_b}  P$, and the rest follows by an argument analogous to the previous case, finishing
the proof of Claim~\ref{claim:union}.\end{proof}

We resume the proof of Lemma \ref{short}. Now,
we construct an accepting run of 
$\A_{\window_0}^\gets$ on $\word$ as follows.
For each integer $ i < 0$, we let $\window_i = (\varrho_i,H_i,T_i,b_i)$ be
 $\window_i' \cup \dots \cup \window_{i+k}'$ with the exception that
$b_i =1$ if there exists $j \in \{i, \dots, k \}$ such that $H'_i \neq T'_i$ and $b_i =0$ otherwise. 
 
Now, we will show that $\window_0, \window_{-1}, \dots$ is indeed an accepting run of $\A_{\window_0}^\gets$ on $\word$.
To this end, we observe that each $\window_{i}$ is a window locally satisfying $\Prog$, 
since so is ${ \window_i' \cup \dots \cup \window_{i+k}' }$ by Claim~\ref{claim:union}; furthermore, each $\window_i$
 it is of the same length as $\window_0$.
Therefore, each $\window_i$ is a state of the automaton
$\A_{\window_0}^\gets$.
Moreover, the definition of $\window_i$ ensures that if there exists
a transition on $\sigma_{i-1}$ from $\window_i'$
to $\window_{i-1}'$ in $\A_{\window_0^L}^\gets$, then there exists a
transition on $\sigma_{i-1}$ from $\window_i$ to $\window_{i-1}$ in
$\A_{\window_0}^\gets$, so $\window_0, \window_{-1}, \dots$ is a run of $\A_{\window_0}^\gets$ on $\word$.
Furthermore, if $\window_i'$ belongs to some set $S$ in the accepting condition of $\A_{\window_0^L}^\gets$, then the bigger window $\window_i$ belongs to the set corresponding to $S$ in the accepting condition of $\A_{\window_0}^\gets$.
Hence, $\A_{\window_0}^\gets$ accepts the run $\window_0, \window_{-1}, \dots$.
The case of $\A_{\window_0}^\to$ is symmetric,  so we obtain that if $\A_{\window_0^R}^\to$ accepts a word $\word$, then so does $\A_{\window_0}^\to$.

To finish the proof of Statement 1, it remains to consider the case  $X =\B$.
Recall that the only difference between the automata of the types $\A$ and $\B$ is that the latter impose additional requirement on the windows to be states (see \Cref{defBC}), namely for a window $\window=(\varrho,H,T,b)$ to be a state, it is required that $H=T$.
This, however, does  not affect our argumentation above, and so, we obtain that $\B_{\window_0}^\gets$ and $\B_{\window_0}^\to$ accept the same words as 
$\B_{\window_0^{L}}^\gets$  and $\B_{\window_0^{R}}^\to$, respectively.

To show that Statements 2 and 3 hold, we recall that the automata $\C$ differ 
from automata $\A$ only in the existence of an additional set $S = \{ (\varrho,H, T,b) \in \states  \mid b=1\}$
in their accepting conditions
(see \Cref{defBC}).
If $\window_0 = (\varrho_0,H_0,T_0,b_0)$ is such that $H_0 = T_0$, 
then we can show that $\C_{\window_0}^\gets$ and $\C_{\window_0}^\to$ accept the same words as 
$\C_{\window_0^{L}}^\gets$  and $\C_{\window_0^{R}}^\to$, respectively, using the same
argumentation as in the proof of Statement 1, except for the following difference.
We use the fact that
$H_0 = T_0$ to ensure that in an arbitrary accepting run of $\C_{\window_0}^\gets$ for some word,
there exists a window $\window_i=(\varrho_i,H_i,T_i,b_i)$ to the left of $\window_0$
with $H_i \neq T_i$, which in turn allows us to ensure that the corresponding run
that we define for $\C_{\window^L_0}^\gets$ contains a ``fragment''
$\window'_j=(\varrho'_j,H'_j,T'_j,b'_j)$ of $\window_i$ with $H'_j\neq T'_j$, and hence
the run is accepting. An analogous argument applies to $\C_{\window_0}^\to$. Thus, Statement~2 holds.

However, if $\window_0 = (\varrho_0,H_0,T_0,b_0)$ is such that $H_0 \neq T_0$---as in Statement 3--- then $b_0=1$ since $\window_0$ is initial, and so, each window in any run of the automata $\C_{\window_0}^\gets$ and $\C_{\window_0}^\to$ will belong to the additional set ${ \{ (\varrho,H, T,b) \in \states  \mid b=1\} }$ from the accepting condition.
Hence, this additional condition is trivially satisfied by any infinite run, 
so $\C_{\window_0}^\gets$ and $\C_{\window_0}^\to$ are equivalent to 
$\A_{\window_0}^\gets$ and $\A_{\window_0}^\to$, respectively.
Therefore, by Statement 1, we obtain that $\C_{\window_0}^\gets$ and $\C_{\window_0}^\to$
are also equivalent to 
$\A_{\window_0^{L}}^\gets$ and $\A_{\window_0^{R}}^\to$, respectively, which proves Statement 3.
\end{proof}

\Cref{short} allows us to restrict our attention to automata with
windows of polynomial length, but representations of such windows
can still be unbounded, because arbitrarily large integers may occur in them. 
However, we will show in \Cref{lem:equiv_auto}
that we can construct equivalent automata with states that can be represented in polynomial space,
by merging ``similar'' states in the original automata.
We use a notion of similarity given by the relation $\sim$ defined next:

\begin{definition}\label{def:equiv_auto}
Let $\sim$ be the equivalence relation on windows such that $\window \sim \window'$ if and only if 
$\window'$ is obtained by increasing all time points mentioned in $\window$ by some integer $c$. 
Let $[\window]_\sim$ be the equivalence class of $\window$ with respect to $\sim$.
Then, for any ${X \in \{ \A,\B,\C \}}$ and any initial
window $\window_0$,
if automaton ${X_{\window_0^L}^\gets = (\states, \Sigma, \delta, q_0,
\acc)}$ is well-defined, then ${\widetilde{X}_{[\window_0^{L}]_\sim}^\gets = (\widetilde{\states},\widetilde{\Sigma}, \widetilde{\delta}, \widetilde{q_0},
\widetilde{\acc})}$ is the nondeterministic B\"{u}chi automaton
defined as follows:
\begin{enumerate}[leftmargin=.3in]
\item $\widetilde{\states}$ is the quotient set of $\states$ by the relation $\sim$,
\item $\widetilde{\Sigma} = \Sigma$,
\item $\widetilde{\delta}$ is such that, for any $q,q' \in \widetilde{\states}$ and $\sigma \in \Sigma$, it holds that $q' \in \widetilde{\delta}(q,\sigma)$ if and only if there exist 
$\window,\window' \in \states$ such that
$[\window]_\sim = q$,
$[\window']_\sim = q'$, and
$q' \in \delta(q,\sigma)$, 
\item $\widetilde{q_0} = [q_0]_\sim$,
\item $\widetilde{\acc}=\{ S_1', \dots, S_n' \}$ is a family of subsets of $\widetilde{\states}$ 
such that $q \in S_i'$ if and only if there exists $\window \in S_i$ such that $[\window]_\sim = q$;
\end{enumerate}
and if ${X_{\window_0^R}^\to = (\states, \Sigma, \delta, q_0,
\acc)}$ is well-defined, then automaton $\widetilde{X}_{[\window_0^{R}]_\sim}^\to$ is defined analogously.
\end{definition}

Now we show that the new automata  ${\widetilde{X}_{[\window_0^{L}]_\sim}^\gets }$ and 
${\widetilde{X}_{[\window_0^{L}]_\sim}^\to }$ are equivalent to the previously defined automata $X_{\window_0^L}^\gets$
an $X_{\window_0^L}^\to$, respectively.

\begin{lemma}\label{lem:equiv_auto}
Let $X \in \{ \A,\B,\C \}$ and
let $\window_0 = (\varrho_0,H_0,T_0,b_0)$
be an initial window locally satisfying $\Prog$  such that $X_{\window_0}^\gets$ and $X_{\window_0}^\to$ are well-defined.
Then, $X_{\window_0^{L}}^\gets$  and $X_{\window_0^{R}}^\to$
are equivalent to
$\widetilde{X}_{[\window_0^{L}]_\sim}^\gets$  and $\widetilde{X}_{[\window_0^{R}]_\sim}^\to$,
respectively.
\end{lemma}
\begin{proof}
Assume that $\window_0, \window_{-1}, \dots$ is an accepting run of $X_{\window_0^L}^\gets$ on a word $\word$.
Then,  by Definition \ref{def:equiv_auto}, $[\window_0]_\sim, [\window_{-1}]_\sim, \dots$ are states of $\widetilde{X}_{[\window_0^{L}]_\sim}^\gets$ and $[\window_0]_\sim = [\window_0^{L}]_\sim$.
Moreover, by the definition of the transition function and the accepting conditions of $\widetilde{X}_{[\window_0^{L}]_\sim}^\gets$, the fact that
$\window_0, \window_{-1}, \dots$ is an accepting run of $X_{\window_0^L}^\gets$  on $\word$ implies that 
$[\window_0]_\sim, [\window_{-1}]_\sim, \dots$ is an accepting run of
$\widetilde{X}_{[\window_0^{L}]_\sim}^\gets$ on the same word $\word$.
Similarly, we can show that if $X_{\window_0^R}^\to$ accepts $\word$, then so does
$\widetilde{X}_{[\window_0^{R}]_\sim}^\to$.

For the opposite direction, let us assume that 
$q_0, q_{-1}, \dots$ is an accepting run of $\widetilde{X}_{[\window_0^{L}]_\sim}^\gets$ on a word $\word$.
For each integer $i < 0$, we let $\window_i = (\varrho_i,H_i,T_i,b_i)$ be a state of $X_{\window_0^L}^\gets$
such that $[\window_i]_\sim = q_i$ and 
${\varrho_i = [\varrho_0^- + i, \varrho_0^+ + i]}$;
such states are guaranteed to exist by definition of 
$X_{[\window_0^L]_{\sim}}^\gets$.
We show that $\window_0, \window_{-1}, \dots$ is an accepting run of $X_{\window_0^L}^\gets$ on $\word$.
To this end, we  will show that for every integer $i \leq 0$ and every symbol $\sigma$, if 
$q_{i-1} \in \widetilde{\delta}(q_i,\sigma)$, then
$\window_{i-1} \in \delta (\window_i, \sigma)$, where $\widetilde{\delta}$ and $\delta$ are the transition functions of $\widetilde{X}_{[\window_0^{L}]_\sim}^\gets$ and 
$X_{\window_0^L}^\gets$, respectively.
Indeed, by definition, if ${ q_{i-1} \in \widetilde{\delta}(q_i,\sigma) }$, then
there are states $\window$ and $ \window' $ of
$X_{\window_0^L}^\gets$ such that
$[\window]_\sim = q_i$, $[\window']_\sim = q_{i-1}$, and $\window' \in \delta (\window, \sigma)$.
Hence, $\window_i \sim \window$ and $\window_{i-1} \sim \window'$.
This,  by the definition of $\sim$ and the fact that $\window' \in \delta (\window, \sigma)$, implies  that 
$\window_{i-1} \in \delta (\window_i, \sigma)$.
Therefore, $\window_0, \window_{-1}, \dots$ is a  run of $X_{\window_0^L}^\gets$ on $\word$.
Finally, since
$q_0, q_{-1}, \dots$ is an accepting run of $\widetilde{X}_{[\window_0^{L}]_\sim}^\gets$,
by the definition of accepting conditions, 
we obtain that 
$\window_0, \window_{-1}, \dots$ is an accepting  run of $X_{\window_0^L}^\gets$.
The same argumentation shows that if 
$\widetilde{X}_{[\window_0^{R}]_\sim}^\to$  accepts $\word$, then $X_{\window_0^R}^\to$
also does so.
\end{proof}

We are now ready to show \EXPS{} data complexity of reasoning in \MTLneg{} over $\Z$.

\begin{theorem}\label{decidability}
Checking whether a  \MTLneg{} program and a 
dataset have a stable model over the integer timeline is in \EXPS{} with respect to data complexity.
\end{theorem}
\begin{proof}
It suffices to show that  checking
existence of  $\window_0$, $\wordL$, and $\wordR$ from \Cref{stable_reduction} is feasible in \EXPS{} in the size of (the representation of) $\D$.
First, we observe that the length of $\varrho_{\PD}$ is  exponential, and  so is the representation  of
windows over $\varrho_0 = \varrho_{\PD}$.
Thus, it is feasible in \EXPS{} to guess a tuple ${\window_0=(\varrho_0,T_0, T_0, 0)}$  as well as  
to verify that it is an initial 
window locally satisfying $\Prog$ and such that   ${T_0 \models \D}$, as required in \Cref{stable_reduction}.

Thus, it remains to check existence of $\wordL$ and $\wordR$ satisfying Conditions 1 and 2 from \Cref{stable_reduction}.
To this end, we will treat a candidate pair of words 
${\word_1=\sigma_{-1} \sigma_{-2} \cdots}$ and ${\word_2=\sigma_1 \sigma_2 \cdots}$ as a single word $\word_1\word_2 = {(\sigma_{-1}, \sigma_{1}) ( \sigma_{-2}, \sigma_2) \cdots}$ and combine pairs of corresponding automata so that they accept combined words.
In particular,
for automata $X$ and $Y$, we let $XY$  be an  automaton
which on a word $\word_1\word_2$ simulates  the run of $X$ on $\word_1$ and the run of  $Y$ on $\word_2$, simultaneously.
Such an automaton $XY$ is polynomially bigger than $X$ and $Y$ as its states are pairs consisting of a state of $X$ and a state of $Y$.
Thus, checking existence of $\wordL$ and $\wordR$
from \Cref{stable_reduction} reduces to checking
existence of $\wordL\wordR$ such that
\begin{itemize}[leftmargin=.3in]
\item[(i)] $\wordL\wordR$ is accepted by $\B^\gets_{\window_0} \B^\to_{\window_0}$ (Condition~1)
and
\item[(ii)] there exists no $\window_0'$ (satisfying properties from Condition~2) such that 
$\wordL\wordR$ is accepted by $\C^\gets_{\window_0'} \A^\to_{\window_0'}$ or by
$\A^\gets_{\window_0'} \C^\to_{\window_0'}$.
\end{itemize}
By \Cref{short,lem:equiv_auto},  Item (i)
reduces to checking if $\wordL\wordR$ is accepted by 
$\widetilde{\B}_{[\window_0^{L}]_\sim}^\gets \widetilde{\B}_{[\window_0^{R}]_\sim}^\to$
and Item (ii) reduces to checking non-existence of $\window_0'=(\varrho_0,H_0,T_0,b)$ (satisfying properties from Condition~2) such that $\wordL\wordR$ is accepted by the union of the automata
$\widetilde{\C}_{[(\window_0')^{L}]_\sim}^\gets \widetilde{\A}_{[(\window_0)'^{R}]_\sim}^\to$ and $\widetilde{\C}_{[(\window_0)'^{L}]_\sim}^\gets \widetilde{\A}_{[(\window_0')^{R}]_\sim}^\to$, if $H_0=T_0$; or by the automaton $\widetilde{\A}_{[(\window_0')^{L}]_\sim}^\gets \widetilde{\A}_{[(\window_0')^{R}]_\sim}^\to$, if $H_0 \neq T_0$.

To perform these checks in \EXPS{}, we will
characterise 
each initial window of the form ${\window'_0=(\varrho_0,H_0, T_0, b_0)}$ locally satisfying $\Prog$ and such that $H_0 \models \D$---as mentioned in Condition 2 of \Cref{stable_reduction}---by a triple consisting of
two initial windows $(\window'_{0})^L$, $(\window'_{0})^R$, and a flag ${b \in \{0,1\}}$ such that $b=1$ if and only if $H_0 \neq T_0$ (so several windows can be characterised by the same triple).
Now, we will show that  a set $\mathcal{T}$ of all  triples characterising all such windows  $\window_0'$ can be constructed in \EXPS{}.
We observe that all $(\window'_{0})^L$ are over the same interval $\varrho^L$, which is independent from the choice of the window  $\window_0'$, since all $\window_0'$ are over the same interval $\varrho_{\PD}$.
Similarly, $(\window'_{0})^R$ are over the same interval $\varrho^R$, which is independent from the choice of  $\window_0'$.
Moreover, these  $\varrho^L$ and $\varrho^R$ are polynomially long, and so, each $(\window'_{0})^L$ and 
$(\window'_{0})^R$ has a polynomial representation.
Thus, there are exponentially many triples  $(\window_1,\window_2,b)$ which need to be checked for 
membership in $\mathcal{T}$.
Our \EXPS{} procedure iterates through all these triples. For each triple 
$(\window_1,\window_2,b)$,
the procedure checks  if there  exists an initial window
${\window'_0=(\varrho_0,H_0, T_0, b_0)}$ locally satisfying $\Prog$
such that $H_0 \models \D$, $(\window'_{0})^L=\window_1$,
$(\window'_{0})^R=\window_2$, and 
$b=1$ if and only if $H_0 \neq T_0$.
All these checks are feasible in \EXPS{}, and if they yield positive answers, then  $(\window_1,\window_2,b) \in \mathcal{T}$.

We observe that checking Items (i) and (ii) reduces to checking  non-emptiness
of an automaton 
obtained by intersecting $\widetilde{\B}_{[\window_0^{L}]_\sim}^\gets \widetilde{\B}_{[\window_0^{R}]_\sim}^\to$
with complements of all the automata corresponding to triples from $\mathcal{T}$,
where we say that a triple ${( (\window'_{0})^L, (\window'_{0})^R,  b) \in \mathcal{T}}$ \emph{corresponds} to the automaton which is  
the intersection of 
$\widetilde{\C}_{[(\window_0')^{L}]_\sim}^\gets \widetilde{\A}_{[(\window_0)'^{R}]_\sim}^\to$ and $\widetilde{\C}_{[(\window_0)'^{L}]_\sim}^\gets \widetilde{\A}_{[(\window_0')^{R}]_\sim}^\to$, if $b=0$; and which is $\widetilde{\A}_{[(\window_0')^{L}]_\sim}^\gets \widetilde{\A}_{[(\window_0')^{R}]_\sim}^\to$, if $b=1$.
To conclude, we  show that this check is feasible in \EXPS{}.
Indeed, by construction, all automata mentioned above have states that can be represented in polynomial space.
However, since these automata are nondeterministic, their complements have states of exponential size.
To intersect the exponentially many such complemented automata,
we construct an automaton whose states are exponentially long tuples, whose $i$th element is a state of the $i$th 
complemented automata.
Such an exponentially long tuple of exponentially large  states is itself  exponentially representable.
Thus, the obtained automaton has  exponentially big states.
Checking non-emptiness of this  automaton 
is feasible in \EXPS{} 
using the standard on-the-fly approach, where states are guessed one-by-one \cite{baier2008principles}.
\end{proof}

The procedure outlined in the proof of the theorem shows that reasoning in \MTLneg{} over the integer timeline is in \EXPS{} in data complexity. 
In the next section, we show how restricting the form of \MTLneg{} programs to the forward-propagating fragment allows us to establish a tight \PS{} bound for data complexity.

\subsection{Forward-Propagating Programs} \label{forward}

In this section, we  consider reasoning with forward-propagating 
programs (see \Cref{rule}) and bounded datasets.
This setting has already been studied in the context of stream reasoning
 \cite{ronca2017stream,WalegaAAAI}.

The normalisation of a \MTLfp{} program,
as defined in \Cref{general}, results also in a \MTLfp{} program; thus,
for the remainder of this section, we let $\Prog$ be a fixed (but arbitrary)
\MTLfp{} program in normal form and let $\D$ be a bounded dataset.

Restricting ourselves to forward-propagating programs and 
bounded datasets will enable a simplification
of the procedure
in \Cref{general},
where we check existence of a stable model of $\Prog$ and $\D$
by looking for an initial window $\window_0$ and a pair of words $\wordL$, $\wordR$
satisfying the conditions in \Cref{stable_reduction}.
Towards this goal, we first show in \Cref{forward_automaton} that
we can guess the initial window $\window_0$
over an interval located to the left of
all intervals in~$\D$, instead of $\varrho_{\PD}$ as in the previous section.
Furthermore, we can use the fact that $\Prog$ is forward-propagating to show that
checking existence of word $\wordL$ can be done independently of $\D$. 
Finally, to check existence of word $\wordR$,
we define a new family
of automata of the form $\F_{\window_0}^\rightarrow$, which can be used 
instead of $X_{\window_0}^\rightarrow$, for $X \in \{\A,\B,\C\}$; doing so
requires no complementation, and thus we 
avoid the exponential blowup. 
As a result, the procedure becomes feasible in polynomial space.

We next define a new automaton $\F_{\window_0}^\rightarrow$
as a refinement of $\B_{\window_0}^\rightarrow$.
Furthermore, we impose an additional restriction on 
states of $\F_{\window_0}^\rightarrow$ to guarantee their minimality, as described below. 

\begin{definition}\label{def:automatonF}
We say that a window $(\varrho,H,T,b)$ \emph{locally satisfies} $\D$
if $M @ t \in H$
for each  ${\matA \in \mat(\Prog,\D)}$
and each $t \in \varrho$ such that 
${\D \models M@t}$.
Let $\window_0=(\varrho_0,T_0,T_0,0)$ be an initial window locally satisfying $\Prog$ (see Definition \ref{def::locallySatisfies}) and $\D$.
We define the generalised B\"{u}chi \mbox{automaton} $\F_{\window_0}^\rightarrow$ analogously to the automaton $\B_{\window_0}^\rightarrow$ in \Cref{defBC},
except that:
\begin{itemize}[leftmargin=.3in]
\item[--] states in $\F_{\window_0}^\rightarrow$ are additionally required to locally satisfy $\D$, and 
\item[--] the transition function is additionally restricted so that
transitions to a window $(\varrho,T,T,0)$ are only allowed if 
there exists no window $(\varrho,H,T,1)$ locally satisfying $\Prog$ and $\D$ such that
$H$ and $T$ coincide over $[\varrho^{-},\varrho^{+}-1]$
and 
$M@\varrho^+ \in T \backslash H$ for some relational atom $M$.
\end{itemize}
\end{definition}

Note that $\F_{\window_0}^\rightarrow$ is essentially deterministic since $\delta(\window,\sigma)$ contains at most one window for every state $\window$ 
and $\sigma \in \Sigma$. 
The following lemma provides the result analogous to \Cref{stable_reduction} for the
setting considered in this section. 

\begin{lemma}\label{forward_automaton}
Program $\Prog$ and dataset $\D$ have a stable model if and only if there exists
an initial window ${\window_0=(\varrho_0,T_0, T_0, 0)}$
locally satisfying $\Prog$ and $\D$ with 
$\varrho_0= [t^{\min}_{\D}- (t_\Prog +1), t^{\min}_{\D}-1]$, which mentions
only constants and predicates from $\Prog$, 
and there exist words $\wordL$ and $\wordR$ over $2^{\mat(\Prog, {\D})}$ such that all of the following conditions
hold:
\begin{enumerate}[leftmargin=.3in]
\item $\wordL$ and $\wordR$ are accepted by  $\B_{\window_0}^\gets$ and $\F_{\window_0}^\to$, respectively,
\item there is no initial window 
${\window'_0=(\varrho_0,H_0, T_0, b_0)}$ locally satisfying $\Prog$ and $\D$ such that $\wordL$ is accepted by 
 $\C_{\window'_0}^\gets$,
\item $\wordL$ mentions only constants and predicates from $\Prog$.
\end{enumerate}
\end{lemma}
\begin{proof}
Assume that $\Tmod$ is a
stable model of $\Prog$ and $\D$. 
To construct the required $\window_0$, $\wordL$, and
$\wordR$,
we let $\dots, \window_{-1}, \window_0, \window_1, \dots$ be the ${[t^{\min}_{\D} - (t_\Prog +1), t^{\min}_{\D}-1]}$-decomposition
of HT-interpretation $(\Tmod,\Tmod)$, and we let  ${\window_i= (\varrho_i, H_i, T_i, b_i)}$. 
Moreover, we let ${\wordL= \sigma_{-1} \sigma_{-2} \cdots}$ and ${\wordR= \sigma_{1} \sigma_{2} \cdots}$ be the words
such that $\sigma_k = T_k \setminus T_{k+1}$ if $k<0$, and $T_k \setminus T_{k-1} $ if ${k>0}$.
In what follows, we will show that the above defined $\window_0$, $\wordL$, and
$\wordR$ satisfy the requirements from the lemma.
First, we observe that
$H_i=T_i$ and $b_i=0$, for each $i \in \mathbb{Z}$. 
Furthermore,
by Lemma \ref{lem:decomp}, each $\window_i$ is a window locally satisfying $\Prog$ and,
since $\Tmod$ is a stable model of $\Prog$ and $\D$, we obtain that each $\window_i$
locally satisfies $\D$.
Therefore, as required in the lemma,
$\window_0$ is an initial window  of the form $(\varrho_0,T_0, T_0, 0)$
with 
$\varrho_0= [t^{\min}_{\D}- (t_\Prog +1), t^{\min}_{\D}- 1]$, and it
locally satisfies  $\Prog$ and $\D$.
Thus,  it remains to show that $\window_0$ mentions only constants and predicates from $\Prog$ and that 
$\window_0$, $\wordL$, and $\wordR$ satisfy Conditions 1--3.

To show that Condition 1 holds, we observe that, by Lemma \ref{lem:decomp}, $\window_0,\window_{-1},\dots$ is
an accepting run of $\A_{\window_0}^\gets$; moreover, since $H_i=T_i$ and $b_i=0$ for all integers $i \leq 0$,
this run is also accepting for $\B_{\window_0}^\gets$. 
Analogously, 
$\window_0,\window_{1},\dots$ is an accepting run of $\B_{\window_0}^\to$.
Now, suppose towards a contradiction
that $\F_{\window_0}^\rightarrow$ does not accept $\wordR$.
Then, by Definition \ref{def:automatonF},  
there exists a smallest integer $i > 0$ such that
$\window_0, \dots , \window_{i-1}$ is a run of
$\F_{\window_0}^\rightarrow$ on $\sigma_1 \cdots \sigma_{i-1}$ (where this word is empty if $i=1$),
but
 $\F_{\window_0}^\rightarrow$
does not have a transition from $\window_{i-1}$ to $\window_{i}$ on $\sigma_{i}$.
Next, we will show how to define windows $\window'_i,\window'_{i+1}, \dots$ such that
$\window_0, \dots, \window_{i-1}, \window'_i, \window'_{i+1},  \dots$ is  an accepting run of 
$\C_{\window_0}^{\to}$ on $\wordR$, which will allow us to raise a contradiction.

\begin{claim}\label{claim}
There exist windows $\window'_i,\window'_{i+1}, \dots$ 
 locally satisfying $\D$  
such that
$\window_0, \dots, \window_{i-1}, \window'_i, \window'_{i+1},  \dots$ is  an accepting run of 
$\C_{\window_0}^{\to}$ on $\wordR$.
\end{claim}
\begin{proof}[Proof of \Cref{claim}]
Recall that $\window_{i} = (\varrho_i,T_i,T_i,0)$. 
Since $\F_{\window_0}^\rightarrow$
does not have a transition from $\window_{i-1}$ to $\window_{i}$ on $\sigma_{i}$, there exists a window ${\window' =(\varrho_i,H',T_i,1)}$ locally satisfying $\Prog$ and $\D$ such that $H'$ coincides with $T_i$ over  $[\varrho^-_i,\varrho^+_{i-1}]$,
and there exists an atom $M' \in \mat{(\Prog,\D)}$ satisfying ${M'@\varrho_i^+ \in T_i \backslash H' }$.
We note also that since $\window'$ is a window 
and  
$H' \subseteq T_i$, there exists an interpretation $\Hmod'$ such that  $\Hmod' \subseteq \Tmod$ and both items from \Cref{def::window} hold.
We will use $\Hmod'$ to define $\window'_i,\window'_{i+1}, \dots$.
To this end, we will say that an  HT-interpretation $(\Hmod^*,\Tmod^*)$ 
\emph{satisfies} a dataset $\D$ and program $\Prog$ \emph{over} an interval $\varrho^*$ if and only if
\begin{itemize}[leftmargin=.3in]
\item[--] for every  $M@\varrho \in \D$
and 
$t \in \varrho \cap \varrho^*$
we have
 $\Hmod^*,t \models M$, and
\item[--] $(\Hmod^*,\Tmod^*)$ satisfy the conditions of Definition \ref{def::HT} for each $t \in \varrho^*$.
\end{itemize}
We can use a construction similar to that in the proof of Theorem \ref{thm::least}
to
define the least interpretation
$\Hmod$ 
such that $\Hmod' \subseteq \Hmod \subseteq \Tmod$,
$\Hmod$ coincides with $\Hmod'$ over  $(-\infty,\varrho_i^-)$, and 
$(\Hmod,\Tmod)$ satisfies $\D$ and $\Prog$ over $[\varrho_i^{-}, \infty)$.
An important observation is that since $\Prog$ is forward-propagating and $\window'$ locally satisfies $\Prog$ 
and $\D$, interpretations $\Hmod'$ and $\Hmod$ also coincide over the interval $\varrho_i$.

Now, we let
$\dots, \window'_{-1}, \window'_0, \window'_1, \dots$, with $\window'_j =  (\varrho'_j,H'_j,T'_j,b'_j)$, 
be the $\varrho_0$-decomposition of $(\Hmod,\Tmod)$;
hence, $\varrho'_j = \varrho_j$ and $T'_j = T_j$, for each $j$.
We observe also that since $\Hmod$ coincides with $\Hmod'$ 
over $(-\infty,\varrho_i^+]$, $H'_i$ agrees with $H'$ on all
relational atoms over $(-\infty,\varrho_i^+]$, so $M'@t \in T_i \backslash H'_i$,
and hence $b'_j=1$ for each $j \geq i$.

We will show that $\window_0, \dots, \window_{i-1}, \window'_i, \window'_{i+1},  \dots$ is an accepting run of $\A_{\window_0}^{\to}$ on
$\wordR$ where all windows locally satisfy $\D$, and then, we will show that this run is also accepted by $\C_{\window_0}^{\to}$.
To this end, we observe that since $\window_0, \dots , \window_{i-1}$ is a run of
$\F_{\window_0}^\rightarrow$ on $\sigma_1 \cdots \sigma_{i-1}$, it is also a run of $\A_{\window_0}^\rightarrow$ on $\sigma_1 \cdots \sigma_{i-1}$.
Moreover,  $\window'_i$ is an initial window,
(since $H'_i \neq T_i$ and $b'_i=1$) and, by an argument analogous to the proof of  \Cref{lem:decomp}, 
each $\window'_j$ for $j \geq i$ locally satisfies $\Prog$ and $\D$, and $\window'_j , \window'_{j+1}, \dots$
is an accepting run of $\A_{\window'_i}^{\to}$ on $\sigma_{i+1}, \sigma_{i+2}, \dots$.
Hence, to show that
$\window_0, \dots, \window_{i-1}, \window'_i, \window'_{i+1},  \dots$ is an accepting run of 
$\A^{\to}_{\window_0}$, it remains to  show that 
$\A^{\to}_{\window_{0}}$ has a transition on $\sigma_{i}$ from $\window_{i-1}$ to $\window'_i$.
To do this, we first show that 
$H'_i$ coincides with $H'$ over  $[\varrho_i^{-},\varrho_{i-1}^{+}]$.
Indeed, consider an arbitrary fact $M@t$ with $M \in \mat{(\Prog,\D)}$ 
and $t \in [\varrho_i^{-},\varrho_{i-1}^{+}]$;  we will show that $M@t \in H'_i$
if and only if $M@t \in H'$.
Note that since $\Prog$ is forward-propagating, any atom 
$M \in \mat{(\Prog,\D)}$ is either of the form $\boxplus_{[0,\infty)} P(\cbf)$
for some relational fact $P(\cbf)$, or it does not mention future operators.
Thus, if $M$ does not mention future operators, the biconditional holds because $\Hmod$ and $\Hmod'$ 
coincide over the interval $(-\infty,\varrho_i^+]$.
Now, suppose $M$ is of the form $\boxplus_{[0,\infty)} P(\cbf)$. 
If $\boxplus_{[0,\infty)} P(\cbf) @ t \in H'$,
then $\boxplus_{[0,\infty)} P(\cbf) @ t \in H'_i$ since $H'\subseteq H'_i$ by construction of $H'_i$.
Conversely, if 
$\boxplus_{[0,\infty)} P(\cbf) @ t \in H'_i$,
we have that $\Hmod,t \models \boxplus_{[0,\infty)} P(\cbf) $,
and since $\Hmod \subseteq \Tmod$, we have $\Tmod,t \models \boxplus_{[0,\infty)} P(\cbf) $,
so $\boxplus_{[0,\infty)} P(\cbf) @ t \in T_i$. Then, since $H'$ agrees with $T_i$ over $[\varrho_i^{-},\varrho_{i-1}^{+}]$,
we conclude that $\boxplus_{[0,\infty)} P(\cbf) @ t \in H'$. Thus, the biconditional
also holds in this case. 
With the above result, we can show that there is a transition
in $\A^{\to}_{\window_{0}}$ from $\window_{i-1}$ to $\window'_i$; we can see this by checking
that all the four conditions from Item 3 in \Cref{def::Buchi} hold.
The first, third, and fourth conditions hold directly by the definition of $\window'_i$
and the fact that $\varrho'_i = \varrho_i$, $T'_i = T_i$, and $b'_i=1$.
The second condition states that $H'_i$ and $H_{i-1}$ must coincide over the
interval $\varrho_{i-1} \cap \varrho_i$. To see this, recall that
$H'_i$ coincides with $H'$ over this interval. In turn,
 $H'$ coincides with $T_i$ over this interval by definition,
and $T_i$ coincides with $T_{i-1}$ over this interval by construction, but 
$T_{i-1}=H_{i-1}$, so $H'_i$ coincides with $H_{i-1}$ and the condition holds.
Therefore, $\window_0, \dots, \window_{i-1}, \window'_i, \window'_{i+1},  \dots$ is an accepting run
of $\A_{\window_0}^{\to}$ on $\wordR$ where all windows locally satisfy $\D$.

Furthermore, $\window_0, \dots, \window_{i-1}, \window'_i, \window'_{i+1},  \dots$ is also an accepting run of 
 $\C_{\window_0}^{\to}$ on $\wordR$, since every window in the run after $\window'_i$ 
has $1$ as its fourth component, and so the additional accepting condition of  $\C_{\window_0}^{\to}$
is satisfied. 
This concludes the proof of Claim \ref{claim}.
\end{proof}

Having proved the claim, we resume the proof of Lemma \ref{forward_automaton}.
Recall that $\window_0,\window_{-1},\dots$ is an accepting run of $\A_{\window_0}^\gets$ on $\wordL$ and, as we showed in the claim above, $\window_0, \dots, \window_{i-1}, \window'_i, \window'_{i+1},  \dots$ is  an accepting run of 
$\C_{\window_0}^{\to}$ on $\wordR$.
Hence, as in the proof of \Cref{stable_reduction}, we can 
use these runs to construct
an HT-model $(\Hmod'',\Tmod)$ of $\Prog$ and $\D$ such that $\Hmod'' \subsetneq \Tmod$; in particular,
$\Hmod'' \models \D$,  as $\varrho_0^{-}$ is to the left of any integer mentioned in $\D$ and 
each window in the run $\window_0, \dots, \window_{i-1}, \window'_i, \window'_{i+1},  \dots$ locally satisfies $\D$.
This, however, means that  $\Tmod$ is not a stable model of
$\Prog$ and $\D$, which raises a contradiction.

To prove Condition 2, suppose towards a contradiction that there exists an initial window
${\window'_0=(\varrho_0, H_0,T_0,b_0)}$ which locally satisfies $\Prog$ and $\D$, and such that
there exits an accepting run $\window'_0, \window'_{-1}, \dots$ 
of $\C_{\window'_0}^{\gets}$ on $\wordL$. We observe that, by definition of $\wordL$, each $\window'_i$ is of the form
$(\varrho_i,H_i,T_i,b_i)$, for some set $H_i$ of metric atoms  and $b_i \in \{0,1\}$.
We consider the least model $\Hmod'$ of all relational facts in $\bigcup_{i \leq 0} H_i$.
The accepting conditions of $\C_{\window'_0}^{\gets}$ ensure that there is an atom $P(\cbf)$ and a time point
$t \in (-\infty,\varrho_0^+]$ such that $\Hmod' ,t \not \models P(\cbf)$ but $\Tmod ,t \models P(\cbf)$. 
By an argument analogous to that in the proof of Theorem \ref{stable_reduction}, and using the facts that
$\window'_0$ locally satisfies $\D$ and $\window'_i$ locally satisfies $\Prog$ for each $i \leq 0$,  
we obtain that $(\Hmod',\Tmod)$ satisfies $\D$ and $\Prog$ over $(-\infty,\varrho_0^+]$.
Then, using a construction similar to that in the proof of Theorem \ref{thm::least},
we can extend $\Hmod'$ to the minimal interpretation $\Hmod$ such that
 $(\Hmod,\Tmod)$ is
an HT-model of $\Prog$ and $\D$. 
Since $\Prog$ is forward-propagating, and $(\Hmod',\Tmod)$ already satisfies $\D$ and $\Prog$
over $(-\infty,\varrho_0^+]$, we have that $\Hmod$ and $\Hmod'$ 
agree over $(-\infty,\varrho_0^+]$.
But then, $\Hmod \subsetneq \Tmod$ since, as we have shown, 
there is $t \in (-\infty,\varrho_0^+]$ such that ${ \Hmod' ,t \not \models P(\cbf) }$
but ${ \Tmod ,t \models P(\cbf) }$.
Thus, $\Tmod$ is not stable, which raises a contradiction.

To prove Condition 3, 
we define a  sequence $\Hmod_0, \Hmod_1,\dots$
of interpretations  as follows:
\begin{itemize}[leftmargin=.3in]
\item[--] $\Hmod_0$ is the least model of $\D$, 
\item[-- ] $\Hmod_\alpha$,  for a successor ordinal $\alpha$,  is the least interpretation such that
for each rule of Form~\eqref{eq:ruleneg} in $\groundp{\Prog}$, and for each time point $t$,   
if ${\Hmod_{\alpha-1},t \models \matA_i}$ for each $1 \leq i \leq k$ and 
${\Tmod,t \not\models \matA_j}$ for each $k+1 \leq j \leq m$,
then ${\Hmod_{\alpha},t \models M}$,
\item[--] $\Hmod_{\alpha} $, for a limit ordinal $\alpha$, is $\bigcup_{\beta < \alpha} 
\Hmod_\beta$. 
\end{itemize}
By the proof of \Cref{thm::least}, each $\Hmod_\alpha$ is well defined and $\Tmod = \Hmod_{\omega_1}$.
Hence, to prove Condition~3, it suffices to show by transfinite induction that, for every ordinal $\alpha$, 
if $\Hmod_\alpha,t \models P(\cbf)$ for some relational atom $P(\cbf)$ and
$t < t^{\min}_{\D}$, then $P$ and all constants in $\cbf$ occur in $\Prog$.
In the base case $\Hmod_0$ is the least model of $\D$. Since  $\D$ is bounded, all the facts it mentions are over intervals contained in $[t^{\min}_{\D}, \infty)$, and so, the statement holds.
In the inductive step for a successor ordinal $\alpha$, we suppose towards a contradiction that $\Hmod_\alpha, t \models P(\cbf)$ and $\Hmod_{\alpha-1}, t \not\models P(\cbf)$, for
some relational atom $P(\cbf)$ and $t < t^{\min}_{\D}$, such that  $P$ or 
some constant in $\cbf$ does not occur in $\Prog$.
Hence, there exists a rule $r$ in $\ground{\Prog}{\D}$
whose  head is $P(\mathbf{c})$, whose positive body atoms are satisfied by $\Hmod_{\alpha-1}$ at $t$,
and whose negated body atoms are satisfied by $\Tmod$ at $t$.
Thus, $P$ appears in $\Prog$, so there exists a constant $c$ in $\mathbf{c}$ which
does not appear in $\Prog$. 
By the safety and the normal form of $r$, the constant $c$ needs to occur in a relational atom $M$ such that either $M$ or $\boxminus_\varrho M$, or $M' \So_\varrho M$ is a positive body atom in $r$.
Each of these cases, however, implies that $\Hmod_{\alpha-1}, t' \models M$, for some $t' \leq t$, which  violates the induction hypothesis.
In the inductive step for a limit ordinal  $\alpha$,
we have $\Hmod_{\alpha} = \bigcup_{\beta < \alpha} 
\Hmod_\beta$; since the claim holds for each $\Hmod_\beta$ by induction hypothesis,
it holds also for $\Hmod_{\alpha}$.

We observe that the above result not only shows that Condition 3 holds, but also that $\window_0$ mentions only constants and predicates from $\Prog$, 
which completes the proof of the first implication from the lemma.

\bigskip  

To show the reverse implication from \Cref{forward_automaton}, let
$\window_0 = (\varrho_0,T_0,T_0,0)$, $\wordL$, and $\wordR$ be as described in \Cref{forward_automaton}.
In particular, by Condition 2, $\B_{\window_0}^\gets$ has an accepting run $\window_0, \window_{-1}, \dots$ on $\wordL$ and
$\F_{\window_0}^\to$ has an accepting run $\window_0, \window_{1}, \dots$ on $\wordR$, where we let
${\window_i= (\varrho_i, T_i, T_i, 0)}$. We will show that  the least model  $\Tmod$ of relational
facts in $\bigcup_{i \in \Z} T_i$ is a stable model of $\Prog$ and $\D$. Using an argument analogous to
the proof of \Cref{stable_reduction}, we obtain that $(\Tmod,\Tmod)$ is an HT-model of $\Prog$ and $\D$.
Now, suppose towards a contradiction that
$\Tmod$ is not a stable model, so there is an interpretation ${\Hmod \subsetneq \Tmod}$	such that
$(\Hmod, \Tmod)$ is an HT-model of $\Prog$ and $\D$. 
We let $\dots, \window'_{-1}, \window'_0, \window'_1, \dots$ be the 
${[t^{\min}_{\D}- (t_\Prog +1), t^{\min}_{\D}-1]}$-decomposition of $(\Hmod,\Tmod)$, where we let
${\window_i'= (\varrho_i', H_i', T_i', b_i')}$.
By construction, $\varrho'_i = \varrho_i$ and $T'_i=T_i$.
Moreover, by Definition \ref{def:decomposition}, $\window'_0$ is an initial window, and it is straightforward to verify
that $\window_i'$ locally satisfies $\Prog$ and $\D$, for each $i \in \Z$.
Now, by Lemma \ref{lem:decomp}, 
$\window'_0, \window'_{-1}, \dots$ is an accepting run of $\A_{\window'_0}^\gets$ on $\wordL$, and  
$\window'_0, \window'_{1}, \dots$ is an accepting run of $\A_{\window'_0}^\to$ on $\wordR$.
Since $\Hmod \subsetneq  \Tmod$, there exists $i \in \mathbb{Z}$ such that 
$\Tmod, \varrho_i^+  \models M$ but $\Hmod, \varrho_i^+  \not \models M$ for some 
relational atom $M \in \mat{(\Prog,\D)}$, which implies $H_i \neq T_i$.
If $i \leq 0$, then $b'_j =1$ for all $j \leq i$, and so
$\C_{\window'_0}^\gets$ accepts $\wordL$, which contradicts Condition~2.
Otherwise, let $i$ be the least (positive) integer such that  
$\Tmod, \varrho_i^+  \models M$ but $\Hmod, \varrho_i^+  \not \models M$, for some 
relational atom $M \in \mat{(\Prog,\D)}$. Observe that this implies $H'_j=T_j$ for each $j <i$, and
$M@ \varrho_i^+  \in T_i \backslash H'_i$. 
By construction, $H'_i$ coincides with $H'_{i-1}$ over  $[\varrho_i^-,\varrho_i^+-1]$;
however, since $H'_{i-1}=T_{i-1}$, and $T_{i-1}$ coincides with $T_i$ over $[\varrho_i^-,\varrho_i^+-1]$,
we have that $H'_i$ coincides with $T_i$ over this interval. 
To sum up, recall that $\window_i$ is of the form $(\varrho_i, T_i, T_i, 0)$,
whereas, by construction, $\window_i'$ is of the form $(\varrho_i, H_i', T_i, 1)$;
furthermore, $H_i'$ coincides with $T_i$ over $[\varrho_i^{-},\varrho_i^{+}-1]$,
and there exists a relational atom $M$ with
$M@\varrho^+ \in T_i \backslash H_i'$.
This, by \Cref{def:automatonF}, implies that the automaton
$\F_{\window_0}^\to$ cannot have a transition to $\window_i$, and so,
$\window_0, \window_{1}, \dots$ is not an accepting run of $\F_{\window_0}^\to$ on  $\wordR$, which raises a contradiction.
\end{proof}

Next, we use \Cref{forward_automaton} to establish a tight \PS{} bound for reasoning in \MTLfp{}.

\begin{theorem}\label{pspace}
Checking whether a  \MTLfp{} program and a bounded dataset have a stable model is \PS{}-complete
with respect to data complexity. 
\end{theorem}

\begin{proof} For the lower bound, we observe that \citeN{walega2020datalogmtl} showed
\PS{}-hardness in data complexity of checking existence of  models for a class of programs which is
strictly smaller than the class of  positive \MTLfp{} programs. Their reduction can be modified in a
straightforward way so that the involved dataset is bounded. Then, \Cref{thm::positive} directly
implies that the same lower bound holds for all (i.e., not necessarily positive) \MTLfp{} programs,
as required.

For the upper bound, by \Cref{forward_automaton}, it suffices to show that checking existence of
a window $\window_0$ and words $\wordL$ and $\wordR$ satisfying the properties described in the statement of the lemma is feasible in \PS{}. First, we observe that  $\window_0$ is over ${\varrho_0= [t^{\min}_{\D}-
(t_\Prog +1), t^{\min}_{\D}-1]}$,  so its length does not depend on $\D$. Hence, $\window_0$ is
polynomially large (in the size of the representation of $\D$), and so it can be guessed in \PS{}; moreover, one can check in \PS{} whether $\window_0$ locally satisfies $\Prog$ and $\D$,
and whether it mentions only constants and predicates from $\Prog$. 
Next, we show how to verify existence of words $\wordL$ and $\wordR$ over $2^{\mat{(\Prog,\D)}}$ such that $\window_0$, $\wordL$, and $\wordR$ satisfy Conditions 1--3.
To verify existence of a
word $\wordL$ accepted by $\B_{\window_0}^\gets$ (first part of Condition 1) which is not accepted
by any  $\C_{\window'_0}^\gets$ (Condition~2) we can use the approach from the proof of
\Cref{decidability}. We observe that  $\window_0$ mentions only  constants and predicates from
$\Prog$, and so, the same holds for windows $\window_0'$ from Condition 2. Moreover, by Condition 3,
 $\wordL$ also mentions  only  constants and predicates from $\Prog$, and so the above check can be
performed independently of $\D$.

It remains to be shown that checking existence of a word $\wordR$ accepted by $\F_{\window_0}^\to$
(that is, the second part of Condition~1) is feasible in \PS{}. To this end, we  check existence of an accepting
run  $\window_0, \window_1, \dots$ of $\F_{\window_0}^\to$ in two steps. First, we  guess windows
$\window_1, \dots , \window_j$ one by one, where ${j= t^{\max}_{\D} - t^{\min}_{\D} + 2 t_{\Prog} +2}$, and
second, we check if $\F_{\window_j}^\to$ has an accepting run. We observe that each of the windows
$\window_1, \dots , \window_j$ is of polynomial size, and so guessing them one by one, as well as
checking that $\F_{\window_0}^\to$ has transitions between consecutive windows, is feasible in
\PS{}. 
To check non-emptiness of the language of $\F_{\window_j}^\to$, we construct for it an
automaton $\widetilde\F_{\window_j}^\to$ in a similar way as we constructed
$\widetilde{X}_{\window_j}^\to$ for $X_{\window_j}^\to$ in the proof of \Cref{decidability}. The 
difference, however, is that the set of states of  $\widetilde\F_{\window_j}^\to$ is the quotient
set of $\sim$ between only those states of $\F_{\window_j}^\to$ whose first elements are intervals located
entirely to the right of $t^{\max}_{\D} + t_\Prog$.
For any such state ${\window_i = (\varrho_i,T_i,T_i,0)}$, checking whether it locally satisfies $\D$ simply amounts to verifying that $M@t \in T_i$ for each $M \in   \mat{(\Prog,\D)}$ such that $M@\varrho_j^{-} \in T_j$, and each $t \in \varrho_i$.
Indeed, this follows from two observations: first,
for any atom $M\in \mat{(\Prog,\D)}$ and two arbitrary time points $t,t' > t^{\max}_{\D} + t_\Prog$, 
$\D \models M @t$ if and only if $\D \models M @t'$; second, $\varrho_j^- > t^{\max}_{\D} + t_\Prog$, by definition of $j$.
Then, using an argument analogous to the proof of Lemma \ref{short}, we can show that $\F_{\window_j}^\to$ and $\widetilde\F_{\window_j}^\to$ are
equivalent. 
Hence, it remains to check if the language of the latter automaton is  non-empty.
By construction, each state of the automaton $\widetilde\F_{\window_j}^\to$
can be represented in polynomial space, so the non-emptiness check is feasible in \PS{} using a standard on-the-fly approach.
\end{proof} 

The assumption that $\D$ is bounded has  been used to ensure existence of a time point
such that no fact of $\D$ holds to the left of it. Thus, our results can be extended to
show that reasoning is still \PS{} in data complexity for datasets where intervals are only bounded
on the left. Furthermore,  none of our results in this section depend on the
direction of time. Indeed, we can define the \emph{backward-propagating} fragment of $\MTLneg$
(analogously to $\MTLfp{}$) as the set of programs where operators $\diamondminus$, $\boxminus$,
and $\So$ are disallowed in rule bodies, and operator $\boxplus$ is disallowed in the head.
Then, we can obtain an analogous set of results for the backward-propagating fragment and
show that reasoning in such fragment is also \PS{} in data complexity.

\section{Related Work}\label{relatedwork}

Positive \MTL{} \cite{datalogMTL} has been  studied over both the  rational \cite{datalogMTL} and the integer 
\cite{walega2020datalogmtl} timelines.
In both cases, the main reasoning tasks are
\EXPS{}-complete for combined complexity \cite{datalogMTL}  and
\PS{}-complete for data complexity \cite{DBLP:conf/ijcai/WalegaGKK19}. Low complexity
fragments  \cite{walegatractable,walkega2021finitely} and  alternative semantics \cite{ryzhikov2019data} have also
been studied.
Practical reasoning algorithms for positive \MTL{} 
have been recently
proposed and implemented in the MeTeoR system \cite{wangAAAI2022}.

$\MTLneg$ is an extension of $\MTL$; therefore, it also extends other prominent
temporal  rule languages captured by $\MTL{}$ such as 
$\mbox{Datalog}_{1S}$   
\cite{chomicki1988temporal,chomicki1989relational} and Templog~\cite{abadi1989temporal}. 
In turn,  $\MTLfp$ generalises the forward-propagating fragment
of $\MTL$ introduced by \citeN{DBLP:conf/ijcai/WalegaGKK19}, and it is thus
related to other forward-propagating temporal logics 
proposed in the literature \cite{baldor2012monitoring,roncaKR,basin2018algorithms}.
Our stable model semantics  \cite{DBLP:conf/kr/WalegaCKG21}  extends
the semantics for stratified \MTLneg{} programs
\cite{tena2021stratified}.

Our approach  is closely related to a recently proposed family of non-monotonic
temporal logics which simultaneously support ASP and modal temporal operators.
Temporal equilibrium logic  (TEL)
\cite{DBLP:conf/eurocast/CabalarV07,aguado2013temporal,DBLP:journals/tplp/CabalarKSS18}
combines propositional ASP with operators from linear temporal logic
\cite{DBLP:conf/focs/Pnueli77}. 
Metric equilibrium logic (MEL) \cite{cabalar2020towards}
extends TEL with metric temporal operators 
that are roughly equivalent to our past operators
$\So_{[0,k]}$ and $\boxminus_{k}$, and their future counterparts.
Both logics introduce non-monotonic semantics for negation based
on stable models, which are defined---as in our work---analogously to 
the models of equilibrium logic.
However, TEL and MEL differ from our approach.
First, they allow formulas supporting all Boolean connectives; therefore,
they can represent disjunction between propositions,
as well as `existential' formulas using  diamond operators; in contrast, in \MTLneg~the 
use of logical connectives and temporal operators is restricted so that all formulas are shaped as rules
similar in spirit to those of Datalog (see Definition \ref{rule}).
Second, they are propositional logics, so 
they do not allow universally quantified variables. 
Finally, the semantics of TEL and MEL are defined on integer timelines with
a least time point, 
whereas we consider both  the full integer timeline as well as the dense
rational timeline. 
In terms of complexity, checking whether a formula 
has a stable model is known to be \PS-hard for both TEL and MEL, and feasible in \EXPS~for TEL.
 
Our approach is also related to the LARS language \cite{DBLP:journals/ai/BeckDE18} for stream reasoning.
LARS is also a rule-based language allowing for negation interpreted under 
stable model semantics.
The differences between  $\MTLneg$ and LARS are as follows.
First, conjuncts in the body of LARS rules
are formulas constructed using temporal operators and unrestricted combinations of
Boolean connectives; in contrast,
body conjuncts in  $\MTLneg$ are metric atoms, which do not mention Boolean operators.
Second,  LARS does not allow for metric operators (only LTL-style boxes and diamonds),
 but it allows for window
operators that have no counterpart in $\MTLneg$. Third,
 LARS rules are meant to be interpreted
at individual time points, so the notion of a stable model in LARS is always relative to a time point (e.g., 
a LARS interpretation can be a stable model of a program at $t$, but not 
at $t+1$); in contrast, a stable model of a $\MTLneg$ program satisfies each rule at 
\emph{every} time point in the timeline (see Definition \ref{def::HT}).
Finally,  LARS interpretations are defined over bounded intervals of the integer timeline,
whereas $\MTLneg$ interpretations are defined over the full integer or rational timeline.
Checking whether a LARS program and dataset have a stable model at a time point $t$ is \PS{}-complete.

The use of ASP for temporal reasoning has been intensively explored, especially in the context of stream reasoning. 
Streamlog \cite{zaniolo2012streamlog} introduces a variant of Datalog with negation 
where atoms are `time-stamped,' in the sense that the first term of an atom is a natural number representing the time point where the atom holds.
Rules in Streamlog must be forward-propagating in time; furthermore, programs must be locally stratified,
which ensures that they have a unique stable model. 
\citeN{DBLP:conf/ai/DoLL11} present an approach which combines a general-purpose ASP solver with
a monotonic stream reasoner to support ASP in stream reasoning; in this approach, both systems
are treated as black boxes. 
Another approach that supports ASP on a temporal setting is \emph{time-decaying reasoning} 
\cite{DBLP:journals/corr/abs-1301-1392}, which relies on 
programs similar to Datalog, but where each
fact and rule need to hold only over a fixed time interval; this behaviour is used to capture changing information
from a sliding window over a sequence of temporal data. 
More recently, negation-free $\MTLneg$ has been applied in the context of monotonic stream reasoning
\cite{WalegaAAAI}, and we see our current work as providing the foundations for extending
this approach with support for ASP.

\section{Conclusion and Future Work}\label{conclusions}

We  extended \MTL{} with negation-as-failure 
under stable model semantics and shown that 
reasoning 
in this language is undecidable over the 
rational timeline but \EXPS~in data complexity over the 
integers.
We  also studied the forward-propagating fragment and shown that, although reasoning 
remains undecidable over the rational timeline, it is becomes \mbox{\PS-complete} in data complexity over the integers  
(thus no harder than in the negation-free case).

We see many avenues for future work. The more immediate challenge
is to provide tight data complexity bounds for reasoning
in the full language over the integer timeline,
where  we currently have an  \EXPS{} upper bound and a \PS{} lower bound for data complexity.
We also plan to consider combined complexity and identify fragments of the language 
where reasoning becomes
decidable over the rational timeline. 
Finally, we are planning to develop practical algorithms and implement
them as an extension the MeTeoR reasoner, that is currently allows for reasoning with positive \MTL{} programs only \cite{wangAAAI2022}.

\section*{Acknowledgments}

This work was funded in whole or in part by
the EPSRC project OASIS (EP/S032347/1),
the EPSRC project UK FIRES (EP/S019111/1),
and the SIRIUS Centre for Scalable Data Access, and
Samsung Research UK.
For the purpose of Open Access, the authors have applied a CC BY public copyright licence to any Author Accepted Manuscript (AAM) version arising from this submission.

\bibliographystyle{acmtrans}
\bibliography{references}

%
%
%


\label{lastpage}
\end{document}